\documentclass[a4paper,fleqn]{cas-sc}

\usepackage[authoryear,longnamesfirst]{natbib}

\usepackage{amsmath,amssymb,amsthm,bm,mathtools}
\usepackage{natbib}
\usepackage{booktabs}
\usepackage{graphicx}
\usepackage{float}
\usepackage{setspace}
\usepackage{microtype}
\usepackage{enumitem}
\usepackage{hyperref}
\usepackage{tabularx}
\usepackage[T1]{fontenc}
\hypersetup{colorlinks=true,linkcolor=black,citecolor=black,urlcolor=black}
\newtheorem{theorem}{Theorem}[section]
\newtheorem{lemma}{Lemma}[section]

\theoremstyle{definition}
\newtheorem{assumption}{Assumption}[section]

\newcommand{\Pp}{\mathbb{P}}
\newcommand{\Ee}{\mathbb{E}}
\newcommand{\Var}{\operatorname{Var}}
\newcommand{\Cov}{\operatorname{Cov}}
\newcommand{\tr}{\operatorname{tr}}
\newcommand{\diag}{\operatorname{diag}}
\newcommand{\rank}{\operatorname{rank}}
\newcommand{\argmin}{\operatorname*{arg\,min}}

\newcommand{\Op}{O_{\mathbb P}}
\newcommand{\R}{\mathbb R}
\newcommand{\ind}{\mathrel{\perp\!\!\!\perp}}
\newcommand{\norm}[1]{\left\lVert #1\right\rVert}
\newcommand{\abs}[1]{\left\lvert #1\right\rvert}

\newcommand{\calM}{\mathcal M}
\newcommand{\ve}{\bm e}
\newcommand{\mA}{\mathbf A}
\newcommand{\mB}{\mathbf B}
\newcommand{\mD}{\mathbf D}

\newcommand{\mI}{\mathbf I}
\newcommand{\mH}{\mathbf H}
\newcommand{\mM}{\mathbf M}
\newcommand{\mP}{\mathbf P}
\newcommand{\mQ}{\mathbf Q}
\newcommand{\mR}{\mathbf R}
\newcommand{\mOmega}{\bm{\Omega}}
\newcommand{\mLambda}{\bm{\Lambda}}
\newcommand{\mS}{\mathbf S}

\newcommand{\mSigma}{\bm{\Sigma}}
\newcommand{\vmu}{\bm\mu}
\newcommand{\vf}{\bm f}
\newcommand{\vY}{\bm Y}
\newcommand{\vZ}{\bm Z}
\newcommand{\vU}{\bm U}
\newcommand{\veps}{\bm\varepsilon}
\newcommand{\vzeta}{\bm\zeta}
\newcommand{\cum}{\operatorname{cum}}
\newcommand{\psinorm}[1]{\left\lVert #1\right\rVert_{\psi_2}}

\newcommand{\vxi}{\bm\xi}
\newcommand{\veta}{\bm\eta}

\def\tsc#1{\csdef{#1}{\textsc{\lowercase{#1}}\xspace}}
\tsc{WGM}
\tsc{QE}
\tsc{EP}
\tsc{PMS}
\tsc{BEC}
\tsc{DE}

\begin{document}
\let\WriteBookmarks\relax
\def\floatpagepagefraction{1}
\def\textpagefraction{.001}
\shorttitle{Factor-Adjusted Location Tests for High-Dimensional Time Series}
\shortauthors{J.Y. Wang et~al.}

\title [mode = title]{Factor-Adjusted Location Tests for High-Dimensional Time Series}                      


\author[1]{Jiyang Wang}

\affiliation[1]{organization={College of Mathematics and Systems Science, Xinjiang University},
	addressline={No. 777, Huarui Street, Shuimogou District}, 
	city={Urumqi},
	postcode={830046}, 
	country={China}}

\author[2]{Xifen Huang}
\cormark[1]

\affiliation[2]{organization={School of Mathematics, Yunnan Normal University},
	addressline={No. 768, Juxian Street, Chenggong District},  
	city={Kunming},
	postcode={650500}, 
	country={China}}
\ead{E-mail address: xf_yellow@126.com}

\author[3]{Long Feng}
\cormark[1]

\affiliation[3]{organization={School of Statistics and Data Science, Nankai University},
	addressline={No. 94, Weijin Road, Nankai District,}, 
	city={Tianjing},
	postcode={300071}, 
	country={China}}

\cortext[cor1]{Corresponding author}
\ead{ flnankai@nankai.edu.cn}

\begin{abstract}
We study high-dimensional one-sample mean testing for time series with strong common serial dependence driven by latent dynamic factors. After estimating the dynamic factor loading space from lagged autocovariance, we project the data onto its orthogonal complement and construct three factor-adjusted tests: a max test for sparse alternatives, a quadratic test for dense alternatives, and a Cauchy combination test for unknown sparsity. The idiosyncratic component is allowed to be non-Gaussian sub-Gaussian vector white noise.  We establish the Gumbel limit of the max statistic, the normal limit and local power function of the quadratic statistic, their asymptotic independence, and the validity of the Cauchy combination.  In the strong-factor case, the refined projection expansion shows that the quadratic statistic remains valid for dimensions as large as $p=o(n^2)$.  A random-loading residual bootstrap is developed for finite-sample calibration.  Simulation studies and a real data application demonstrate reliable size control and competitive power for high-dimensional observations with strong dependence.

\end{abstract}


\begin{highlights}
\item Propose a high-dimensional time series mean test with factor adjustment, which eliminates strong serial correlation caused by latent common factors through dynamic factor projection.
\item Construct Max, Sum, and Cauchy combination tests, respectively suitable for sparse, dense, and alternatives with unknown sparsity.
\item Develop an asymptotic theory for the statistic and propose a residual bootstrap calibration; simulation and empirical results based on the S\&P 500 demonstrate that this method achieves reliable size control and competitive power.
\end{highlights}


\maketitle

\section{Introduction}
\label{sec:intro}

High-dimensional data analysis is now central in statistics, econometrics, finance, genomics, signal processing, and environmental studies.  In such applications, the dimension $p$ is often comparable to, or much larger than, the sample size $n$.  Classical multivariate methods based on fixed-$p$ asymptotics are then poorly calibrated or even undefined.  Testing whether a high-dimensional mean vector is zero is one of the most basic examples. \citet{Hotelling1931}'s $T^2$ statistic is optimal in classical Gaussian models with fixed dimension, but it requires inversion of the sample covariance matrix and is not applicable when $p$ is large relative to $n$.

For independent observations, the literature has developed two complementary families of tests.  Sum-type procedures aggregate many small coordinatewise effects and are powerful under dense alternatives.  Representative contributions include \citet{BaiSaranadasa1996}, \citet{SrivastavaDu2008}, \citet{ChenQin2010}, \citet{ChenPaulPrenticeWang2011}, \citet{SrivastavaKatayamaKano2013}, and \citet{GregoryEtAl2015}.  Robust and nonparametric variants include \citet{WangPengLi2015}, \citet{FengSun2016}, and \citet{ChakrabortyChaudhuri2017}.  Max-type, thresholding, and higher-criticism tests target sparse alternatives with a small number of large coordinates; see, for example, \citet{DonohoJin2004}, \citet{HallJin2010}, \citet{ZhongChenXu2013}, and \citet{CaiLiuXia2014}.  Adaptive methods that combine dense and sparse evidence have also been studied, including \citet{XuLinWeiPan2016}, \citet{LiAuePaulPengWang2020}, \citet{FengJiangLiLiu2024}, and the related max-sum construction of \citet{FengJiangLiuXiong2022}.  Gaussian approximation theory for maxima of high-dimensional sums, such as \citet{ChernozhukovChetverikovKato2013}, provides important distributional tools for this line of work.  A broad review is given by \citet{HuangLiLiYang2022}.

When observations are serially dependent, mean testing becomes more delicate because temporal dependence changes both centering and variance.  Ignoring this dependence may inflate the size of tests designed for independent samples.  Several high-dimensional time-series tests therefore adjust for temporal dependence by long-run covariance estimation, self-normalization, or dependence-aware Gaussian approximation.  \citet{AyyalaParkRoy2017} considered mean testing for high-dimensional dependent observations, with further corrections in \citet{ChoLimAyyalaParkRoy2019}.  \citet{ZhangWu2017} developed Gaussian approximation for high-dimensional stationary time series, \citet{BasuMichailidis2015} studied regularized estimation for sparse high-dimensional time series, and \citet{WangShao2020} proposed self-normalized inference.  Related recent developments include \citet{ZhangChenQiu2025}.  These methods are most natural when the remaining temporal dependence is weak, local, or can be summarized by a stable long-run covariance structure.

In many panel time series, however, the dependence is strong because many coordinates share a few latent dynamic components.  Financial returns may contain market-wide serial dynamics; macroeconomic panels may be driven by business-cycle factors; environmental measurements may share large-scale spatio-temporal patterns.  Dynamic factor models provide a natural way to represent such strong common dependence.  Important contributions include \citet{PenaBox1987}, \citet{ForniHallinLippiReichlin2000}, \citet{StockWatson2002}, \citet{BaiNg2002}, \citet{HallinLiska2007}, \citet{PanYao2008}, \citet{LamYaoBathia2011}, and \citet{LamYao2012}.  In particular, the dynamic factor viewpoint separates a high-dimensional time series into a low-dimensional serially dependent component and a vector white-noise component.  This decomposition is well suited for high-dimensional mean testing: one can first remove the low-rank dynamic component and then test the projected residual mean.

Several factor-adjusted mean tests are especially relevant to this paper.  \citet{ZhouKong2015} studied high-dimensional mean testing under an approximate factor model and used factor adjustment to reduce the influence of common components.  \citet{ZhangZhouHeZhang2018} considered high-dimensional time-series mean testing with a factor structure, constructed adaptive statistics from factor-adjusted data with different norms, and used multiplier bootstrap calibration.  \citet{HeZhangZhangZhou2020} developed factor-adjusted two-sample mean tests and support recovery procedures.  These papers show that factor adjustment is useful for high-dimensional mean inference, but their inferential targets and modeling roles of the factors are different from ours. The first and third papers address approximate factor adjustment for mean testing in one-sample and two-sample settings, respectively; the second paper investigates factor-structured high-dimensional time series, employing adaptive norm aggregation subsequent to factor adjustment.  In contrast, our approach employs a dynamic factor decomposition to model and eliminate strong common serial dependence; the post-projection residual is allowed to have a singular correlation matrix and non-Gaussian sub-Gaussian white-noise distribution; and we jointly analyze the proposed max, sum, and Cauchy-combination test statistics, including max-sum asymptotic independence, local dense-alternative power for the sum statistic, and the refined effect of estimating the dynamic factor space. 

This paper develops factor-adjusted high-dimensional location tests for time series with strong common serial dependence.  We consider
\begin{equation*}
	\vY_t=\vmu+\bm W_t,\qquad \bm W_t=\mA\vf_t+\veps_t,
\end{equation*}
where ${\bm f}_t$ is a low-dimensional dynamic factor, $\mA$ is a loading matrix, and $\veps_t$ is vector white noise.  The loading space is estimated from lagged autocovariances.  The observations are then projected onto the orthogonal complement of the estimated loading space.  On the projected data, we construct a max statistic, a quadratic statistic, and a Cauchy combination statistic.

The contributions are as follows.  First, the proposed testing framework explicitly targets high-dimensional mean inference under strong serial dependence generated by latent dynamic factors.  This differs from methods that assume independent observations or weakly dependent residuals after simple covariance adjustment.  Second, the procedure is adaptive to the sparsity of the alternative: the max statistic is sensitive to sparse signals, the quadratic statistic is sensitive to dense signals, and the Cauchy combination remains effective when the sparsity level is unknown.  Third, the theory proves that estimating and removing the dynamic factor space does not change the first-order null limits of the projected max and quadratic statistics.  The proof does not assume Gaussian residuals; the idiosyncratic component is sub-Gaussian vector white noise.  Fourth, for local dense alternatives, the quadratic statistic has a noncentral normal limit, which gives a closed-form asymptotic power function.  Fifth, the asymptotic independence between the max and quadratic statistics justifies a simple analytic Cauchy calibration.  Finally, we provide a random-loading residual bootstrap that matches the simulation implementation and repeats the factor adjustment within each bootstrap sample for improved finite-sample size control.

The rest of the article is organized as follows.  Section~\ref{sec:model} states the model.  Section~\ref{sec:method} defines the testing procedure.  Section~\ref{sec:theory} presents the core theoretical results.  Section~\ref{sec:simulation} reports finite-sample size and power evidence.  Section~\ref{sec:empirical-sp500-weekly} reports the empirical S\&P 500 weekly excess-return study.  All intermediate lemmas and proofs are in the Appendix.

\section{Model}
\label{sec:model}

We observe a $p$-dimensional time series $\vY_1,\ldots,\vY_n$ and test
\begin{equation}
	H_0:\vmu=\bm0\qquad\hbox{against}\qquad H_1:\vmu\ne\bm0.
	\label{eq:hypothesis}
\end{equation}
The data are generated by
\begin{equation}
	\vY_t=\vmu+\mA\vf_t+\veps_t,\qquad t=1,\ldots,n,
	\label{eq:model}
\end{equation}
where $\mA\in\R^{p\times r}$, $\mA^\top\mA=\mI_r$, $r$ is fixed, $\vf_t\in\R^r$ is a latent dynamic factor, and $\veps_t$ is a vector white-noise process.  The number of factors $r$ and the loading space $\calM(\mA)$ are unknown.  Since $(\mA,\vf_t)$ is identifiable only up to an orthogonal rotation, all inference is expressed through the projection matrices
\begin{equation*}
	\mP=\mA\mA^\top,
	\qquad
	\mQ=\mI_p-\mP.
\end{equation*}
The projection removes the factor component.  It also removes the component of $\vmu$ lying in $\calM(\mA)$.  The main theory is stated under the convention
\begin{equation}
	\mP\vmu=\bm0,
	\label{eq:identifiable_mean}
\end{equation}
so that $\mQ\vmu=\vmu$ and \eqref{eq:hypothesis} is equivalent to the projected hypothesis.  Without \eqref{eq:identifiable_mean}, the same procedure tests $H_0:\mQ\vmu=\bm0$.

Let
\begin{equation*}
	\Gamma_y(k)=\Cov(\vY_{t+k},\vY_t),\qquad k\ge1.
\end{equation*}
For a fixed lag truncation $k_0\ge1$, define the Lam--Yao population matrix
\begin{equation}
	\mM=\sum_{k=1}^{k_0}\Gamma_y(k)\Gamma_y(k)^\top.
	\label{eq:M}
\end{equation}
Under the dynamic factor assumptions of \citet{LamYao2012}, the columns of $\mA$ span the eigenspace of $\mM$ corresponding to its $r$ nonzero eigenvalues.  In the sample, let
\begin{equation*}
	\widehat\Gamma_y(k)=\frac{1}{n-k}\sum_{t=1}^{n-k}(\vY_{t+k}-\bar{\vY})(\vY_t-\bar{\vY})^\top,
	\qquad
	\widehat\mM=\sum_{k=1}^{k_0}\widehat\Gamma_y(k)\widehat\Gamma_y(k)^\top.
\end{equation*}
Let $\widehat\lambda_1\ge\cdots\ge\widehat\lambda_p$ be the eigenvalues of $\widehat\mM$.  We use the ratio rule
\begin{equation*}
	\widehat r=\argmin_{1\le j\le R}\frac{\widehat\lambda_{j+1}}{\widehat\lambda_j},\qquad r<R<p,
\end{equation*}
unless stated otherwise.  Given $\widehat r$, let $\widehat\mA$ be the matrix of the first $\widehat r$ orthonormal eigenvectors of $\widehat\mM$ and define
\begin{equation*}
	\widehat\mP=\widehat\mA\widehat\mA^\top,
	\qquad
	\widehat\mQ=\mI_p-\widehat\mP.
\end{equation*}

\section{Methodology}
\label{sec:method}

The feasible projected observations are
\begin{equation*}
	\widehat{\vU_t}=\widehat\mQ\vY_t,
	\qquad
	\bar{\widehat{\vU}}=n^{-1}\sum_{t=1}^n\widehat{\vU_t}.
\end{equation*}
Let
\begin{equation*}
	\widehat d_i=\frac1n\sum_{t=1}^n(\widehat U_{ti}-\bar{\widehat U}_i)^2,
	\qquad
	\widehat Z_{ti}=\widehat U_{ti}/\widehat d_i^{1/2},
	\qquad
	\bar{\widehat{\vZ}}=n^{-1}\sum_{t=1}^n\widehat{\vZ_t}.
\end{equation*}
The sample correlation matrix of the projected residuals is
\begin{equation*}
	\widehat{\mS}=\frac1n\sum_{t=1}^n(\widehat{\vZ_t}-\bar{\widehat{\vZ}})(\widehat{\vZ_t}-\bar{\widehat{\vZ}})^\top,
	\qquad
	\widehat\mR=\diag(\widehat\mS)^{-1/2}\widehat\mS\diag(\widehat\mS)^{-1/2}.
\end{equation*}

The max statistic is
\begin{equation*}
	T_M=\max_{1\le i\le p}n\bar{\widehat Z}_i^{\,2}-2\log p+\log\log p.
\end{equation*}
It is designed for sparse alternatives where only a few coordinates of $\vmu$ are nonzero.  Its analytic null calibration uses the Gumbel distribution
\begin{equation*}
	F_G(x)=\exp\{-\pi^{-1/2}\exp(-x/2)\}.
\end{equation*}
The sum statistic is
\begin{equation}
	T_S=\frac{n\bar{\widehat{\vZ}}^{\,\top}\bar{\widehat{\vZ}}-\{(n-1)p/(n-3)\}}
	{\{2[\tr(\widehat{\mR}^2)-p^2/(n-1)]\}^{1/2}}.
	\label{eq:TS}
\end{equation}
It is designed for dense alternatives where many small coordinates jointly create a large Euclidean norm.  The finite-sample corrections in \eqref{eq:TS} account for coordinatewise variance estimation and sample-correlation estimation.  Under local dense alternatives, the same statistic has a noncentral normal limit.  Let $\bm\gamma_n=\sqrt n\mD^{-1/2}\vmu$ and $s_p^2=2\tr(\mR^2)$.  The leading noncentrality and variance inflation are
\begin{equation*}
	\Delta_n=\frac{\norm{\bm\gamma_n}_2^2}{s_p},
	\qquad
	\Lambda_n=\frac{2\bm\gamma_n^\top\mR\bm\gamma_n}{\tr(\mR^2)}.
\end{equation*}
Theorem~\ref{thm:power} shows that $T_S$ is asymptotically $N(\Delta,1+\Lambda)$ when $\Delta_n\to\Delta$ and $\Lambda_n\to\Lambda$, yielding the analytic power approximation $1-\Phi\{(z_{1-\alpha}-\Delta)/\sqrt{1+\Lambda}\}$.

Let
\begin{equation*}
	\widehat p_M=1-F_G(T_M),
	\qquad
	\widehat p_S=1-\Phi(T_S),
\end{equation*}
where $\Phi$ is the standard normal distribution function.  The combined Cauchy statistic is
\begin{equation*}
	T_C=\frac12\tan\{\pi(1/2-\widehat p_M)\}+\frac12\tan\{\pi(1/2-\widehat p_S)\}.
\end{equation*}
Its analytic $p$-value is
\begin{equation*}
	\widehat p_C=\frac12-\frac1\pi\arctan(T_C).
\end{equation*}
The Cauchy rule is attractive because it preserves a simple analytic calibration under broad dependence among component $p$-values \citep{LiuXie2020}.

For finite samples, we use the same random-loading residual bootstrap as in the simulations.  Compute
\begin{equation*}
	\widehat{\vf}_t=\widehat\mA^\top(\vY_t-\bar{\vY}),
	\qquad
	\widehat{\veps}_t=\widehat\mQ(\vY_t-\bar{\vY}),
	\qquad
	\widetilde{\veps}_t=\widehat{\veps}_t-n^{-1}\sum_{s=1}^n\widehat{\veps}_s.
\end{equation*}
In the present construction $\sum_t\widehat{\veps}_t=\bm0$, so $\widetilde{\veps}_t=\widehat{\veps}_t$; the explicit centering is kept to define the bootstrap null sample.  For each bootstrap replication $b$, draw an independent random loading matrix $\widetilde{\mA}_b\in\R^{p\times\widehat r}$ with orthonormal columns and independent indices
\begin{equation*}
	I_{1b},\ldots,I_{nb}\stackrel{\mathrm{i.i.d.}}{\sim}\mathrm{Unif}\{1,\ldots,n\}.
\end{equation*}
In implementation $\widetilde{\mA}_b$ is obtained by orthonormalizing a $p\times\widehat r$ matrix with independent standard normal entries.  The bootstrap sample is
\begin{equation}
	\vY_{tb}^*=\widetilde{\mA}_b\widehat{\vf}_t+\widetilde{\veps}_{I_{tb}},
	\qquad t=1,\ldots,n.
	\label{eq:bootstrap_sample}
\end{equation}
For each bootstrap sample, we apply the same factor-estimation, projection, standardization, and testing algorithm to obtain $T_{M,b}^*$, $T_{S,b}^*$, and $T_{C,b}^*$.  The bootstrap $p$-values are
\begin{equation*}
	\widehat p_M^{\rm boot}=\frac{1+\sum_{b=1}^B\mathbf 1\{T_{M,b}^*\ge T_M\}}{B+1},
	\qquad
	\widehat p_S^{\rm boot}=\frac{1+\sum_{b=1}^B\mathbf 1\{T_{S,b}^*\ge T_S\}}{B+1}.
\end{equation*}
The bootstrap Cauchy statistic is formed from $\widehat p_M^{\rm boot}$ and $\widehat p_S^{\rm boot}$.

\section{Theoretical results}
\label{sec:theory}

This section gives only the core assumptions and the core theorems.  All loading-space perturbation, feasible-oracle replacement, and distributional calculations are collected in the Appendix.

Let
\begin{equation*}
	\mOmega=\mQ\mSigma_\varepsilon\mQ,
	\qquad
	\mD=\diag(\mOmega),
	\qquad
	\mR=\mD^{-1/2}\mOmega\mD^{-1/2}.
\end{equation*}
Let $\bm a_i^\top=\ve_i^\top\mA$ be the $i$th row of $\mA$.

\begin{assumption}[Dynamic factor structure]
	\label{ass:model}
	The decomposition \eqref{eq:model} satisfies $\mA^\top\mA=\mI_r$, no nonzero linear combination of $\vf_t$ is white noise, and the identifiable mean condition \eqref{eq:identifiable_mean} holds.  The factor process is centered, $\Ee\vf_t=\bm0$, weakly stationary, uniformly sub-Gaussian after the normalization in Assumption~\ref{ass:LY}, and geometrically strongly mixing.
\end{assumption}

\noindent\textit{Comment.}
The exclusion of white-noise linear combinations of $\vf_t$ is an identification convention for the minimal dynamic-factor representation, rather than a substantive restriction on the data-generating process.  If a nonzero linear combination of the latent process is itself white noise, then its contribution can be absorbed into the idiosyncratic term; the lag-autocovariance matrix $\mM$ identifies only directions contributing to nonzero-lag autocovariances.  Thus $\vf_t$ denotes the dynamically identifiable factor process.  A white-noise common component is covered by the present formulation only when, after absorption into $\veps_t$, Assumption~\ref{ass:resid} remains valid; a strong white-noise common component that creates a spiked $\mSigma_\varepsilon$ would require an additional static-factor adjustment.  This convention follows the dynamic-factor identification logic of \citet{LamYao2012}.

\begin{assumption}[Dynamic factor strength and regularity]
	\label{ass:LY}
	Let $\Sigma_f(k)=\Cov(\vf_{t+k},\vf_t)$ and $W_f=\{\Sigma_f(1),\ldots,\Sigma_f(k_0)\}\in\R^{r\times rk_0}$.  There exists $\delta\in[0,1]$ and constants $0<c_f<C_f<\infty$ such that
	\begin{equation}
		c_fp^{1-\delta}\le \sigma_{\min}(W_f)\le \norm{W_f}\le C_fp^{1-\delta}.
		\label{eq:C5}
	\end{equation}
	The matrix $\mM$ in \eqref{eq:M} has exactly $r$ positive eigenvalues,
	\begin{equation*}
		\nu_1(\mM)>\cdots>\nu_r(\mM)>0=\nu_{r+1}(\mM)=\cdots=\nu_p(\mM).
	\end{equation*}
	Moreover,
	\begin{equation*}
		\max_i\norm{\bm a_i}_2^2\le Cp^{-1},
		\qquad
		\sum_{\ell=-\infty}^{\infty}\norm{\Cov(\vf_{t+\ell},\vf_t)}\le Cp^{1-\delta},
	\end{equation*}
	and
	\begin{equation}
		\sup_t\sup_{\norm{\bm u}_2=1}
		\psinorm{p^{-(1-\delta)/2}\bm u^\top\vf_t}\le C,
		\qquad
		\alpha_f(m)\le C\exp(-cm),
		\label{eq:factor_subg_mix}
	\end{equation}
	where $\alpha_f(m)$ is the strong-mixing coefficient of $\{\vf_t\}$.  The fourth-order factor cumulants are summable in the following uniform sense: for any unit vectors $\bm u_1,\ldots,\bm u_4\in\R^r$,
	\begin{equation}
		\sum_{h_1,h_2,h_3\in\mathbb Z}
		\left|
		\cum(\bm u_1^\top\vf_0,\bm u_2^\top\vf_{h_1},\bm u_3^\top\vf_{h_2},\bm u_4^\top\vf_{h_3})
		\right|
		\le C p^{2-2\delta}.
		\label{eq:factor_cumulant}
	\end{equation}
\end{assumption}

\begin{assumption}[Factor-number selection]
	\label{ass:rhat}
	The factor-number estimator satisfies
	\begin{equation*}
		\Pp(\widehat r=r)\to1.
	\end{equation*}
	A sufficient condition for the Lam--Yao ratio estimator is the sharp drop at $r$ together with post-$r$ separation: for some $a_n\downarrow0$,
	\begin{equation*}
		\widehat\lambda_{r+1}/\widehat\lambda_r=\Op(p^{2\delta}/n),
		\qquad
		p^{2\delta}/n=o(a_n),
		\qquad
		\Pp\left(\min_{r<j\le R}\widehat\lambda_{j+1}/\widehat\lambda_j>a_n\right)\to1.
	\end{equation*}
\end{assumption}

\begin{assumption}[Independent-component idiosyncratic errors]
	\label{ass:resid}
	The idiosyncratic errors are independent of the factor process and are independent over time.  There exists a positive definite covariance matrix $\mSigma_\varepsilon$ and i.i.d. innovation vectors $\vxi_t=(\xi_{t1},\ldots,\xi_{tp})^\top$ such that
	\begin{equation}
		\veps_t=\mSigma_\varepsilon^{1/2}\vxi_t,
		\qquad
		\Ee\vxi_t=\bm0,
		\qquad
		\Cov(\vxi_t)=\mI_p,
		\label{eq:ic_eps_model}
	\end{equation}
	where $\xi_{t1},\ldots,\xi_{tp}$ are mutually independent and
	\begin{equation*}
		\sup_{t,j}\norm{\xi_{tj}}_{\psi_2}\le C.
	\end{equation*}
	The eigenvalues of $\mSigma_\varepsilon$ are uniformly bounded:
	\begin{equation}
		0<c_\varepsilon\le \lambda_{\min}(\mSigma_\varepsilon)
		\le \lambda_{\max}(\mSigma_\varepsilon)\le C_\varepsilon<\infty.
		\label{eq:Sigmae_bounded}
	\end{equation}
	The analytic Gumbel calibration of $T_M$ uses the following extreme-value regularity of the induced oracle correlation matrix $\mR$: for some $0<\varrho<1$,
	\begin{equation}
		\max_{i<j}|R_{ij}|\le \varrho,
		\qquad
		p^{-1}\left|\{i: |\{j\ne i: |R_{ij}|\ge(\log p)^{-2}\}|\ge p^{\kappa_p}\}\right|\to0,
		\label{eq:EV_resid}
	\end{equation}
	where $\kappa_p\downarrow0$ and $p^{\kappa_p}/\log^4p\to\infty$.
\end{assumption}

\begin{assumption}[Simple sufficient growth condition]
	\label{ass:growth}
	The dimension has polynomial growth, $p=O(n^\kappa)$, for a constant $\kappa>0$.  If $\delta=0$, assume
	\begin{equation*}
		0<\kappa<2.
	\end{equation*}
	If $0<\delta\le1$, assume
	\begin{equation*}
		0<\kappa<\min\{2,(1/2+2\delta)^{-1}\}.
	\end{equation*}
	Thus the strong-factor case allows $p=o(n^2)$.  In the weak-factor case the more conservative exponent above makes the plug-in eigenspace remainder negligible.  Polynomial growth implies $\log^m p=o(n)$ for every fixed $m>0$.
\end{assumption}

\begin{theorem}[Null distributions, asymptotic independence, and Cauchy combination]
	\label{thm:null}
	Suppose Assumptions~\ref{ass:model}--\ref{ass:growth} hold.  Under $H_0:\vmu=\bm0$,
	\begin{equation*}
		T_M\Rightarrow G,
		\qquad
		\Pp(G\le x)=\exp\{-\pi^{-1/2}\exp(-x/2)\},
	\end{equation*}
	\begin{equation*}
		T_S\Rightarrow N(0,1),
	\end{equation*}
	and $T_M$ and $T_S$ are asymptotically independent.  Consequently,
	\begin{equation*}
		T_C\Rightarrow C,
		\qquad
		\Pp(C\le x)=\frac12+\frac1\pi\arctan(x),
	\end{equation*}
	so that the asymptotic level-$\alpha$ combined test rejects when $T_C>\tan\{\pi(1/2-\alpha)\}$.
\end{theorem}

\begin{theorem}[Local alternatives and power of the sum statistic]
	\label{thm:power}
	Suppose Assumptions~\ref{ass:model}--\ref{ass:growth} hold and $\mP\vmu=\bm0$.  Set
	\begin{equation*}
		\bm\nu=\mD^{-1/2}\vmu,
		\qquad
		\bm\gamma_n=\sqrt n\,\bm\nu,
		\qquad
		m_\infty=\norm{\bm\gamma_n}_\infty,
	\end{equation*}
	\begin{equation*}
		\Delta_n=\frac{\norm{\bm\gamma_n}_2^2}{\{2\tr(\mR^2)\}^{1/2}},
		\qquad
		\Lambda_n=\frac{2\bm\gamma_n^\top\mR\bm\gamma_n}{\tr(\mR^2)}.
	\end{equation*}
	If
	\begin{equation}
		\Delta_n\to\Delta\in[0,\infty),
		\qquad
		\Lambda_n\to\Lambda\in[0,\infty),
		\label{eq:local_alt_condition}
	\end{equation}
	then, under the local alternative sequence $H_{1n}:\vmu\ne\bm0$,
	\begin{equation*}
		T_S\Rightarrow N(\Delta,1+\Lambda).
	\end{equation*}
	Consequently, for the one-sided level-$\alpha$ sum test rejecting when $T_S>z_{1-\alpha}$,
	\begin{equation*}
		\Pp_{\vmu}(T_S>z_{1-\alpha})
		\to
		1-\Phi\left(\frac{z_{1-\alpha}-\Delta}{\sqrt{1+\Lambda}}\right).
	\end{equation*}
	If
	\begin{equation}
		\frac{\Delta_n}{\sqrt{1+\Lambda_n}}\to\infty,
		\label{eq:sum_consistency_condition}
	\end{equation}
	then the sum test has power tending to one.  In particular, because $\lambda_{\max}(\mR)$ is bounded and $\tr(\mR^2)\asymp p$, the simpler condition $\Delta_n\to\infty$ is sufficient.  The max test has power tending to one if
	\begin{equation*}
		m_\infty^2-2\log p\to\infty.
	\end{equation*}
	The Cauchy test has power tending to one whenever either the max condition or \eqref{eq:sum_consistency_condition} holds.
\end{theorem}

\begin{theorem}[Bootstrap validity]
	\label{thm:bootstrap}
	Assume the conditions of Theorem~\ref{thm:null}.  Consider the random-loading residual bootstrap in \eqref{eq:bootstrap_sample}.  Conditional on the data, suppose the regenerated loading matrices have orthonormal columns and satisfy $\max_i\|\ve_i^\top\widetilde\mA_b\|_2^2=O_{\mathbb P}^*(\log p/p)$ uniformly in $b$.  Suppose further that either $\widehat r^*\equiv\widehat r$ is fixed in the bootstrap algorithm or
	\begin{equation*}
		\Pp^*(\widehat r^*=\widehat r\mid\vY_1,\ldots,\vY_n)\to1
		\qquad\text{in probability}.
	\end{equation*}
	All conditional statements below are understood on this bootstrap selection event.  If the number of bootstrap replications satisfies $B\to\infty$, then
	\begin{equation*}
		\sup_x\left|\Pp^*(T_M^*\le x)-\Pp_0(T_M\le x)\right|\to0,
		\qquad
		\sup_y\left|\Pp^*(T_S^*\le y)-\Pp_0(T_S\le y)\right|\to0
	\end{equation*}
	in probability, where $\Pp_0$ denotes the first-order null law of the corresponding statistic.  Consequently, the bootstrap marginal $p$-values and the Cauchy combination formed from them are first-order valid.
\end{theorem}

\section{Simulation studies}
\label{sec:simulation}

This section reports finite-sample evidence for the proposed factor-adjusted tests.  The simulations have two goals.  The size experiments examine whether factor adjustment and the random-loading residual bootstrap control type-I error when the observations contain strong common serial dependence.  The power experiments examine how the max, sum, and Cauchy components respond to alternatives with different sparsity levels.

\subsection{Design and implementation}
\label{subsec:simulation-design}

The baseline data-generating mechanism is
\begin{equation*}
	\vY_t=\vmu+\mA\vf_t+\veps_t,\qquad t=1,\ldots,n,
\end{equation*}
where the factor component is absent when $r=0$.  When $r=2$, the entries of the initial loading matrix are generated independently from $U(-1,1)$ and then orthonormalized.  The two factors follow independent AR(1) processes
\begin{equation*}
	f_{jt}=\phi_j f_{j,t-1}+u_{jt},\qquad \phi_1=0.6,\qquad \phi_2=-0.5,
\end{equation*}
with Gaussian innovations.  The idiosyncratic errors are generated as
\begin{equation*}
	\veps_t=\mSigma^{1/2}\bm z_t,\qquad \Sigma_{ij}=0.5^{|i-j|}.
\end{equation*}
We consider two innovation distributions for $\bm z_t$: standard Gaussian innovations and standardized $t_5$ innovations.  The $t_5$ variables are divided by $\sqrt{5/3}$ so that each coordinate has unit variance.

The proposed methods are denoted by FA-Max, FA-Sum, and FA-CC.  The factor number is estimated by the lagged-autocovariance eigenvalue-ratio rule with $k_0=5$.  For all proposed tests, finite-sample critical values are obtained by the random-loading residual bootstrap used in the simulations.  After estimating factor scores $\widehat{\vf}_t$ and residuals $\widehat{\bm e}_t$, the $b$th bootstrap sample is generated as
\begin{equation}
	\vY_{tb}^*=\widetilde{\mA}_b\widehat{\vf}_t+\widehat{\bm e}_{I_{tb}},
	\qquad I_{tb}\sim\mathrm{Unif}\{1,\ldots,n\},
	\label{eq:simulation_bootstrap}
\end{equation}
where $\widetilde{\mA}_b$ is a randomly generated $p\times\widehat r$ orthonormal loading matrix.  When $\widehat r=0$, \eqref{eq:simulation_bootstrap} reduces to resampling centered residuals.  Each Monte Carlo design uses 1000 replications, the nominal level is $\alpha=0.05$, and the bootstrap procedures use $B=500$ bootstrap replications.

\subsection{Empirical size}
\label{subsec:size}

For the size study, $\vmu=\bm0$, $n\in\{100,200,400\}$, $p\in\{100,200\}$, and $r\in\{0,2\}$.  The case $r=0$ is a no-factor benchmark.  The case $r=2$ contains strong common serial dependence driven by two dynamic factors.  We compare FA-Max, FA-Sum, and FA-CC with the Srivastava--Du type test (SD), the Wang--Shao self-normalized test (WS), the band-excluded $U$-statistic test (ZCQ), the test of Ayyala--Park--Roy (APR), and the max-type procedure of Zhang--Wu (ZW).  SD, WS, ZCQ, APR and ZW are applied to the raw observations as in their original formulations; in particular, the SD and WS entries below are computed from the original data, not from the factor-adjusted residuals.  At a true size of 0.05, the Monte Carlo standard error is $\{0.05(1-0.05)/1000\}^{1/2}=0.0069$.

\begin{table}[H]
	\centering
	\caption{Empirical size under Gaussian innovations.}
	\label{tab:size-normal}
	\resizebox{\textwidth}{!}{%
		\begin{tabular}{ccccccccccc}
			\toprule
			$n$ & $p$ & $r$ & FA-Max & FA-Sum & FA-CC & SD & WS & ZCQ & APR & ZW \\
			\midrule
			100 & 100 & 0 & 0.055 & 0.074 & 0.076 & 0.037 & 0.047 & 0.073 & 0.048 & 0.125 \\
			100 & 100 & 2 & 0.053 & 0.051 & 0.052 & 0.215 & 0.038 & 0.162 & 0.143 & 0.228 \\
			100 & 200 & 0 & 0.067 & 0.086 & 0.083 & 0.050 & 0.048 & 0.066 & 0.044 & 0.149 \\
			100 & 200 & 2 & 0.043 & 0.064 & 0.049 & 0.172 & 0.048 & 0.154 & 0.128 & 0.236 \\
			200 & 100 & 0 & 0.050 & 0.059 & 0.052 & 0.054 & 0.047 & 0.071 & 0.046 & 0.102 \\
			200 & 100 & 2 & 0.044 & 0.044 & 0.043 & 0.226 & 0.049 & 0.131 & 0.124 & 0.148 \\
			200 & 200 & 0 & 0.066 & 0.075 & 0.068 & 0.048 & 0.046 & 0.066 & 0.049 & 0.138 \\
			200 & 200 & 2 & 0.047 & 0.050 & 0.053 & 0.190 & 0.040 & 0.143 & 0.135 & 0.187 \\
			400 & 100 & 0 & 0.056 & 0.061 & 0.058 & 0.046 & 0.051 & 0.073 & 0.056 & 0.121 \\
			400 & 100 & 2 & 0.043 & 0.059 & 0.054 & 0.207 & 0.054 & 0.098 & 0.091 & 0.127 \\
			400 & 200 & 0 & 0.045 & 0.060 & 0.058 & 0.053 & 0.057 & 0.062 & 0.044 & 0.128 \\
			400 & 200 & 2 & 0.052 & 0.044 & 0.052 & 0.207 & 0.046 & 0.136 & 0.124 & 0.163 \\
			\bottomrule
		\end{tabular}%
	}
\end{table}

\begin{table}[H]
	\centering
	\caption{Empirical size under standardized $t_5$ innovations.}
	\label{tab:size-t5}
	\resizebox{\textwidth}{!}{%
		\begin{tabular}{ccccccccccc}
			\toprule
			$n$ & $p$ & $r$ & FA-Max & FA-Sum & FA-CC & SD & WS & ZCQ & APR & ZW \\
			\midrule
			100 & 100 & 0 & 0.052 & 0.055 & 0.055 & 0.056 & 0.063 & 0.072 & 0.051 & 0.104 \\
			100 & 100 & 2 & 0.040 & 0.047 & 0.050 & 0.206 & 0.040 & 0.140 & 0.123 & 0.223 \\
			100 & 200 & 0 & 0.061 & 0.083 & 0.073 & 0.057 & 0.052 & 0.081 & 0.046 & 0.138 \\
			100 & 200 & 2 & 0.040 & 0.050 & 0.046 & 0.174 & 0.038 & 0.136 & 0.118 & 0.206 \\
			200 & 100 & 0 & 0.052 & 0.061 & 0.061 & 0.045 & 0.047 & 0.073 & 0.049 & 0.114 \\
			200 & 100 & 2 & 0.049 & 0.053 & 0.053 & 0.214 & 0.048 & 0.141 & 0.124 & 0.163 \\
			200 & 200 & 0 & 0.051 & 0.066 & 0.065 & 0.066 & 0.050 & 0.067 & 0.047 & 0.129 \\
			200 & 200 & 2 & 0.041 & 0.054 & 0.046 & 0.198 & 0.045 & 0.150 & 0.129 & 0.172 \\
			400 & 100 & 0 & 0.060 & 0.063 & 0.070 & 0.044 & 0.044 & 0.062 & 0.046 & 0.118 \\
			400 & 100 & 2 & 0.041 & 0.041 & 0.045 & 0.211 & 0.049 & 0.113 & 0.097 & 0.134 \\
			400 & 200 & 0 & 0.059 & 0.068 & 0.069 & 0.049 & 0.057 & 0.063 & 0.050 & 0.123 \\
			400 & 200 & 2 & 0.061 & 0.047 & 0.066 & 0.204 & 0.041 & 0.111 & 0.093 & 0.150 \\
			\bottomrule
		\end{tabular}%
	}
\end{table}

Tables~\ref{tab:size-normal} and \ref{tab:size-t5} show that the proposed factor-adjusted bootstrap tests control type-I error in the strong-dependence designs.  Across the $r=2$ rows, the average sizes of FA-Max, FA-Sum, and FA-CC are 0.046, 0.050, and 0.051 over the Gaussian and $t_5$ experiments, respectively.  Their largest absolute deviations from the nominal level are 0.011, 0.014, and 0.016, which are close to the Monte Carlo error scale.  The raw-data competitors behave differently.  SD is severely oversized, with average size 0.202 in the same $r=2$ settings; ZCQ, APR, and ZW are also liberal, with average sizes 0.135, 0.119, and 0.178.  WS is close to nominal in size, with average size 0.045, but the power experiments below show that it is substantially less sensitive to the alternatives considered here.  For example, under Gaussian innovations with $(n,p,r)=(100,200,2)$, FA-Max, FA-Sum, and FA-CC have sizes 0.043, 0.064, and 0.049, whereas SD, ZCQ, APR, and ZW have sizes 0.172, 0.154, 0.128, and 0.236, respectively.

The $r=0$ rows provide a benchmark in which no common dynamic factor needs to be removed.  In this easier setting, the raw SD and WS procedures are close to nominal, and APR is also close to nominal, whereas ZW remains somewhat liberal.  The contrast between $r=0$ and $r=2$ confirms that the dominant source of size distortion for the liberal raw-data competitors is the unremoved common dynamic component.  The standardized $t_5$ results are similar to the Gaussian results, indicating that the bootstrap calibration is not driven by Gaussianity.

\subsection{Empirical power}
\label{subsec:power}

The power designs use $n=200$, $p=100$, nominal level $0.05$, and fixed Euclidean signal strength
\begin{equation*}
	\norm{\vmu}_2=0.45.
\end{equation*}
The mean vector has $s$ nonzero entries,
\begin{equation*}
	\mu_j=\pm\frac{0.45}{\sqrt{s}},\qquad j=1,\ldots,s,
	\qquad \mu_j=0,\qquad j>s,
\end{equation*}
where the signs are generated independently.  The sparsity level varies over
\begin{equation*}
	s\in\{1,2,3,4,5,6,8,10,12,15,18,22,27,33,40,48,58,70,85,100\}.
\end{equation*}
This design keeps the total $\ell_2$ signal fixed while decreasing the coordinatewise signal as $s$ increases.

\subsubsection*{Dynamic-factor design}

We first consider the same $r=2$ dynamic factor design as in the size experiment.  Figure~\ref{fig:power-r2} reports the empirical powers of FA-Max, FA-Sum, and FA-CC as the number of nonzero mean components varies.

\begin{figure}[h]
	\centering
	\includegraphics[width=0.92\textwidth]{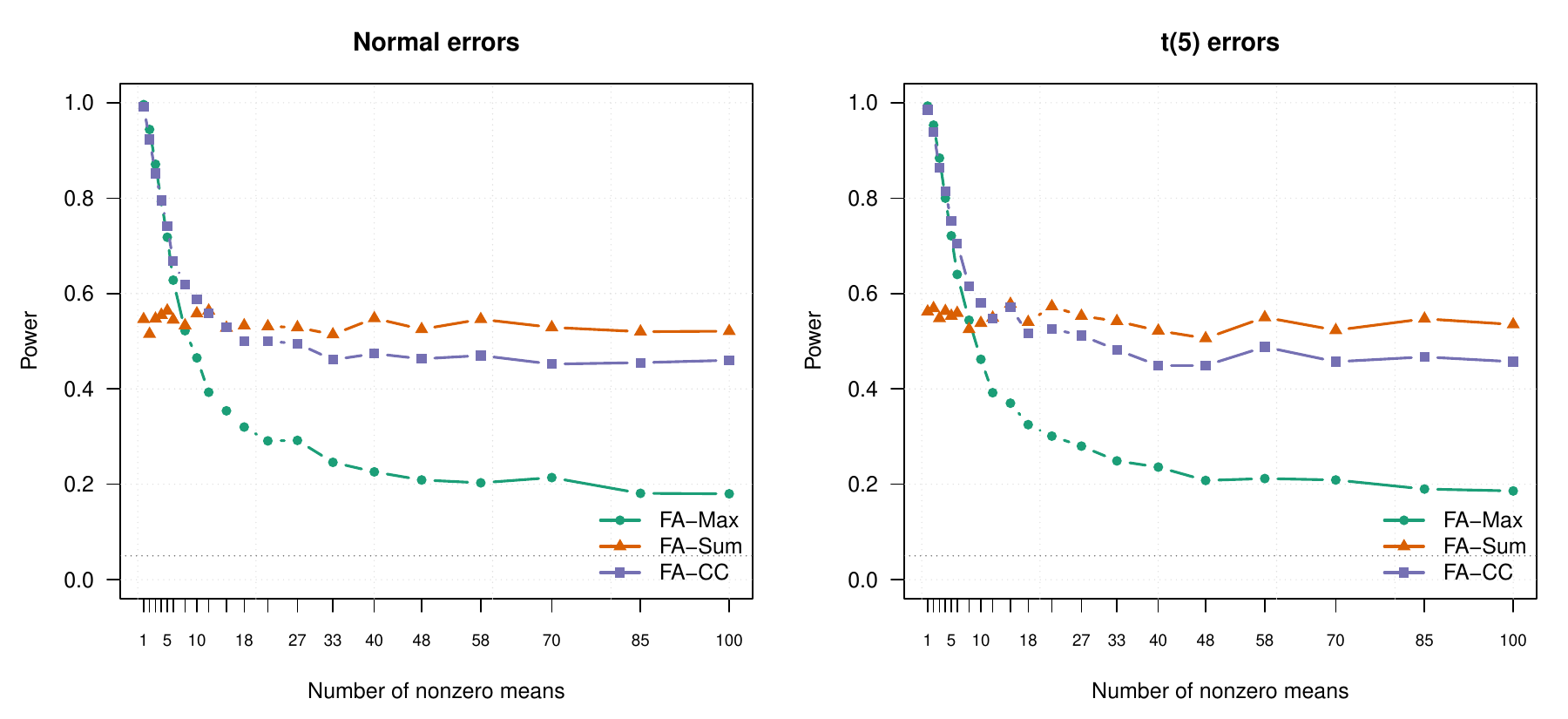}
	\caption{Empirical power in the dynamic-factor design as a function of the number of nonzero mean components.}
	\label{fig:power-r2}
\end{figure}

Figure~\ref{fig:power-r2} displays the expected sparse--dense tradeoff.  FA-Max is the most sensitive component when the alternative is very sparse, because the fixed $\ell_2$ signal is then concentrated in a few coordinates.  Its power decreases as $s$ increases and the coordinatewise signal weakens.  FA-Sum is much more stable over the entire sparsity range and becomes the dominant component in moderately dense and dense regimes.  FA-CC combines these two behaviors: it tracks FA-Max for very sparse alternatives and stays close to the sum component once the signal becomes more diffuse.  The Gaussian and standardized $t_5$ panels have similar shapes, suggesting that the random-loading residual bootstrap is not driven by Gaussianity.

\subsubsection*{No-factor benchmark}

We next set $r=0$ and generate $\vY_t=\vmu+\veps_t$ with the same idiosyncratic covariance $\Sigma_{ij}=0.5^{|i-j|}$.  This design isolates the effect of sparsity when there is no common dynamic factor to estimate and remove.  Figure~\ref{fig:power-r0} reports the resulting power curves; SD, WS and ZCQ are included as raw-data competitors.

\begin{figure}[h]
	\centering
	\includegraphics[width=0.92\textwidth]{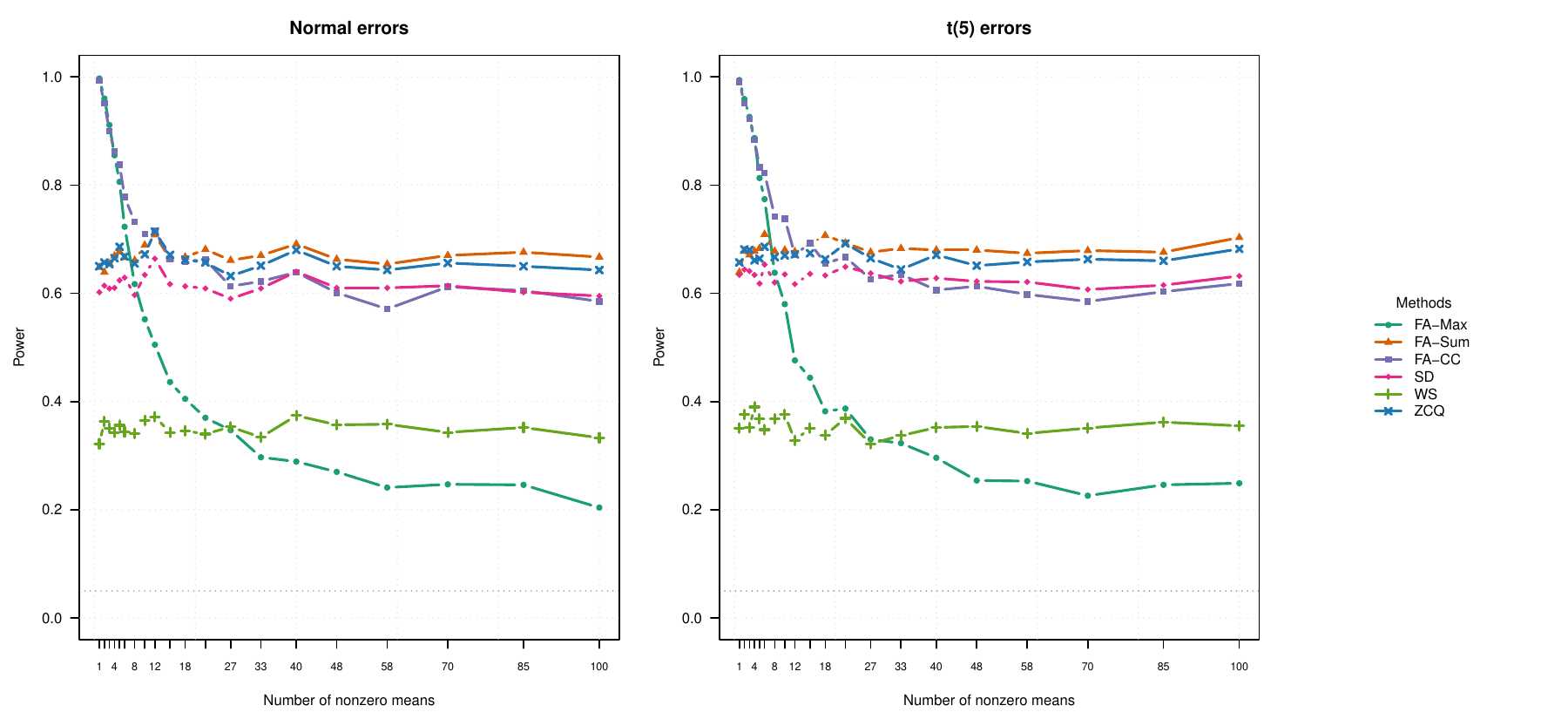}
	\caption{Empirical power in the no-factor benchmark as a function of the number of nonzero mean components.}
	\label{fig:power-r0}
\end{figure}

Figure~\ref{fig:power-r0} confirms that the proposed procedures do not suffer an artificial power loss when no common dynamic factor is present.  FA-Max again dominates under the sparsest alternatives and decreases as the signal is spread over more coordinates.  FA-Sum remains stable across the sparsity range and is competitive with the dense-signal competitors.  FA-CC retains the sparse-signal strength of FA-Max while remaining close to the sum-type methods in moderately dense and dense cases.  Among the benchmark procedures, SD and ZCQ behave similarly to sum-type methods in dense regimes, whereas WS is visibly less powerful throughout most of the range.  The Gaussian and standardized $t_5$ panels are close, reinforcing the robustness of the bootstrap implementation under moderate heavy-tailedness.

\section{Empirical study: weekly excess returns of S\&P 500 stocks}
\label{sec:empirical-sp500-weekly}

We apply the proposed factor-adjusted tests to a high-dimensional panel of weekly excess returns of S\&P 500 stocks.  The data are constructed from the local wide-format daily excess-return file \texttt{sp500\_excess\_return\_wide.csv}.  Each entry is the individual stock return in excess of the risk-free rate.  To reduce the effect of short-run market microstructure noise and to obtain a cleaner time direction for the bootstrap calibration, daily excess returns are aggregated into weekly excess returns by summing within each trading week.

We use the most recent ten-year period available in the data, from January 3, 2014 to December 29, 2023.  Stocks with missing observations over this period are removed.  The final panel has sample size
\[
n=522
\]
and dimension
\[
p=333.
\]
The number of latent dynamic factors is selected by the lagged-autocovariance ratio method described in Section~\ref{sec:method}, which gives
\[
\widehat r=1.
\]
The testing problem is
\[
H_0:\ \vmu=\Ee(\vY_t)=\bm0,
\qquad
H_1:\ \vmu\ne\bm0,
\]
where \(\vY_t\) is the vector of weekly stock excess returns.  Economically, the null states that the expected excess returns of the selected S\&P 500 stocks are jointly zero after removing the latent dynamic component.

\subsection{White-noise diagnostics}

The factor-adjusted procedure is designed to remove common serial dependence before applying high-dimensional mean tests.  We therefore first compare the raw weekly excess returns with the factor-adjusted residuals in terms of their remaining serial correlation.  For each stock, we apply the Ljung--Box test at lag 3 \citep{LjungBox1978}.  Since this gives 333 simultaneous diagnostic tests, we report both the raw rejection proportion and the rejection proportions after the Benjamini--Hochberg adjustment \citep{BenjaminiHochberg1995}.

\begin{table}[htbp]
	\centering
	\setlength{\tabcolsep}{17pt}
	\caption{White-noise diagnostics for weekly excess returns.}
	\label{tab:sp500-weekly-whitenoise}
	\begin{tabular}{lcccc}
		\hline
		Series & Median $p$-value & Raw $p<0.05$ & BH $p<0.05$ & BH $p<0.10$ \\
		\hline
		Weekly excess returns & 0.0782 & 43.24\% & 28.23\% & 40.54\% \\
		Factor-adjusted residuals & 0.2610 & 25.83\% & 3.90\% & 8.41\% \\
		\hline
	\end{tabular}
\end{table}

Table~\ref{tab:sp500-weekly-whitenoise} shows that the raw weekly excess returns still contain substantial serial dependence.  After BH adjustment, 28.23\% of the stock-level Ljung--Box tests reject the white-noise null at the 5\% level, and 40.54\% reject at the 10\% level.  This is precisely the type of setting in which high-dimensional mean tests that treat the observations as independent, or only weakly serially dependent after a direct covariance correction, can be unreliable.

The factor adjustment removes most of this time-series dependence.  For the factor-adjusted residuals, only 3.90\% of the BH-adjusted tests reject at the 5\% level, and 8.41\% reject at the 10\% level.  These proportions are close to the corresponding nominal levels.  Thus the estimated dynamic factor captures the dominant common time dependence, leaving residuals that are much closer to vector white noise.  This diagnostic supports the use of the proposed factor-adjusted testing procedure for the subsequent mean analysis.

\subsection{High-dimensional mean testing results}

We compute the three proposed statistics: FA-Max, FA-Sum, and FA-CC.  The bootstrap calibration uses \(B=500\) repetitions.  The resulting statistics and bootstrap $p$-values are shown in Table~\ref{tab:sp500-weekly-mean}.

\begin{table}[htbp]
	\centering
	\setlength{\tabcolsep}{60pt}
	\caption{Factor-adjusted high-dimensional mean tests for weekly excess returns.}
	\label{tab:sp500-weekly-mean}
	\begin{tabular}{lcc}
		\hline
		Method & Test statistic & $p$-value \\
		\hline
		FA-Max & 7.2579 & 0.0200 \\
		FA-Sum & 1.9329 & 0.0818 \\
		FA-CC  & --     & 0.0322 \\
		\hline
	\end{tabular}
\end{table}

At the 5\% level, FA-Max and FA-CC reject the joint zero-mean null, whereas FA-Sum does not.  The result is economically and statistically informative.  FA-Max is designed to detect a small number of pronounced residual mean components, whereas FA-Sum is most sensitive to broad dense alternatives.  The fact that FA-Max rejects but FA-Sum does not suggests that the evidence against the zero-excess-return benchmark is sparse rather than dense.  In other words, the rejection is driven by a relatively small subset of stocks with large average excess returns after removing the common dynamic component, rather than by a weak mean shift spread over most of the cross-section.

The Cauchy combination also rejects, with $p$-value 0.0322.  This agrees with the proposed interpretation of FA-CC: it inherits the sensitivity of FA-Max under sparse alternatives while retaining protection against dense alternatives through the FA-Sum component.  Therefore, in this empirical example, the combined test provides an adaptive rejection without requiring the analyst to specify the sparsity level of the alternative in advance.

\subsection{Sparsity check via individual tests}

As a complementary diagnostic, we conduct individual two-sided mean tests for each stock and apply the BH procedure to the 333 resulting $p$-values.  The results are reported in Table~\ref{tab:sp500-weekly-sparse}.

\begin{table}[width=1.9\linewidth,cols=5,pos=h]
	\centering
	\setlength{\tabcolsep}{17pt}
	\caption{Individual mean tests with BH adjustment.}
	\label{tab:sp500-weekly-sparse}
	\begin{tabular}{lcccc}
		\hline
		Series & Raw $p<0.05$ & BH $p<0.05$ & BH $p<0.10$ & Median $p$-value \\
		\hline
		Weekly excess returns & 71 (21.32\%) & 2 (0.60\%) & 7 (2.10\%) & 0.1625 \\
		Factor-adjusted residuals & 28 (8.41\%) & 1 (0.30\%) & 3 (0.90\%) & 0.3806 \\
		\hline
	\end{tabular}
\end{table}

The individual-test results are consistent with the global testing conclusions.  Before factor adjustment, 71 stocks have raw individual $p$-values below 5\%, but only 2 remain significant after BH adjustment at the 5\% FDR level.  After factor adjustment, the number of raw rejections decreases to 28, and only 1 stock remains significant at the 5\% FDR level.  At the 10\% FDR level, only 3 factor-adjusted residual means are selected.  Hence the rejection by FA-Max and FA-CC should not be interpreted as evidence for a market-wide dense mean shift.  Rather, it points to a sparse deviation from the zero-excess-return null.

Overall, the empirical study illustrates the role of factor adjustment in financial panels.  The raw returns contain strong common serial dependence, while the factor-adjusted residuals are much closer to white noise.  After removing this common dynamic component, the proposed adaptive test detects a sparse residual mean signal.  This is consistent with the simulation evidence: the max component is effective for sparse alternatives, the sum component targets dense alternatives, and the Cauchy combination provides a robust aggregate decision when the sparsity pattern is unknown.

\section{Conclusion}
\label{sec:conclusion}

This paper proposes factor-adjusted high-dimensional location tests for time series with strong common serial dependence.  The procedure estimates a Lam--Yao dynamic factor structure, projects the data onto the orthogonal complement of the estimated loading space, and then applies max, sum, and Cauchy-combined tests to the projected residuals.  The core theory shows that factor estimation does not alter the first-order null limits under primitive factor-strength and residual calibration conditions.  The simulations show that the proposed random-loading residual bootstrap implementation controls size under strong dynamic dependence.  The power experiments with and without latent factors further show that the Cauchy combination adapts across sparse and dense alternatives, while the no-factor benchmark confirms that factor adjustment does not create artificial power loss when no common dynamic factor is present.  The empirical study of weekly S\&P 500 excess returns further illustrates the practical role of factor adjustment: after removing one dynamic factor, the residuals are much closer to white noise, and the global rejection is driven by a sparse residual mean signal.  Future work will develop fully data-driven bootstrap theory for the factor-number estimator and study robust extensions for heavy-tailed vector white noise.

\section*{CRediT authorship contribution statement}
\textbf{Jiyang Wang:} Methodology, Investigation, Software, Funding acquisition. 
\textbf{Xifen Huang:} Data curation, Investigation, Software, Writing -- original draft preparation. 
\textbf{Long Feng:} Conceptualization, Funding acquisition, Methodology, Project administration, Supervision.

\appendix
\section{Appendix: Intermediate results and proofs}

\subsection{Consequences of the primitive idiosyncratic-error condition}

\begin{lemma}[Residual correlation properties induced by the independent-component model]
	\label{lem:resid_consequences}
	Let Assumptions~\ref{ass:LY} and \ref{ass:resid} hold.  With $\mOmega=\mQ\mSigma_\varepsilon\mQ$, $\mD=\diag(\mOmega)$, and $\mR=\mD^{-1/2}\mOmega\mD^{-1/2}$, the following assertions hold:
	\begin{equation}
		0<c_d\le\min_iD_{ii}\le\max_iD_{ii}\le C_d<\infty,
		\qquad
		\lambda_{\max}(\mR)\le C_R<\infty.
		\label{eq:spectrum_resid}
	\end{equation}
	The matrix $\mR$ is positive semidefinite and may be singular, with $\rank(\mR)=p-r$.  Moreover,
	\begin{equation}
		\tr(\mR^2)\asymp p,
		\qquad
		\tr(\mR^m)\le C_R^{m-1}p\quad(m\ge2),
		\qquad
		\frac{\tr(\mR^4)}{\tr^2(\mR^2)}=O(p^{-1}),
		\label{eq:trace_resid}
	\end{equation}
	\begin{equation}
		\max_{1\le i\le p}\sum_{j=1}^pR_{ij}^2
		=\max_i(\mR^2)_{ii}\le C_R.
		\label{eq:row_l2_resid}
	\end{equation}
	In addition,
	\begin{equation}
		\vzeta_t=\mD^{-1/2}\mQ\veps_t,
		\qquad
		\Ee\vzeta_t=\bm0,
		\qquad
		\Cov(\vzeta_t)=\mR,
		\qquad
		\vzeta_t\ \text{are independent over }t,
		\label{eq:zeta_basic}
	\end{equation}
	\begin{equation}
		\sup_t\sup_{\norm{\bm u}_2=1}\psinorm{\bm u^\top\vzeta_t}\le K_\varepsilon,
		\qquad
		\sup_{t,i}\psinorm{\varepsilon_{ti}/\{\Var(\varepsilon_{ti})\}^{1/2}}\le K_\varepsilon,
		\label{eq:subg_resid}
	\end{equation}
	and there exists a deterministic $p\times p$ matrix $\mB_R$ and a $p$-vector $\veta_t$ with independent coordinates, $\Ee\veta_t=\bm0$, $\Cov(\veta_t)=\mI_p$, and $\sup_{t,j}\norm{\eta_{tj}}_{\psi_2}\le C$, such that
	\begin{equation}
		\mB_R\mB_R^\top=\mR,
		\qquad
		\vzeta_t=\mB_R\veta_t.
		\label{eq:linear_subg_resid}
	\end{equation}
	
\end{lemma}

\begin{proof}
	For $i=1,\ldots,p$, put $\bm a_i^\top=\ve_i^\top\mA$.  Since
	\begin{equation*}
		\mQ\ve_i=\ve_i-\mA\bm a_i,
		\qquad
		\norm{\mQ\ve_i}_2^2=1-\norm{\bm a_i}_2^2,
	\end{equation*}
	Assumption~\ref{ass:LY} gives, for all sufficiently large $p$,
	\begin{equation*}
		1-Cp^{-1}\le \norm{\mQ\ve_i}_2^2\le1.
	\end{equation*}
	Together with \eqref{eq:Sigmae_bounded},
	\begin{equation*}
		c_\varepsilon(1-Cp^{-1})
		\le
		\ve_i^\top\mQ\mSigma_\varepsilon\mQ\ve_i
		\le
		C_\varepsilon,
	\end{equation*}
	which proves the first part of \eqref{eq:spectrum_resid}.  Also,
	\begin{equation*}
		\lambda_{\max}(\mR)
		=\lambda_{\max}(\mD^{-1/2}\mQ\mSigma_\varepsilon\mQ\mD^{-1/2})
		\le \norm{\mD^{-1/2}}^2\norm{\mQ}^2\norm{\mSigma_\varepsilon}
		\le C.
	\end{equation*}
	Since $\mSigma_\varepsilon$ is positive definite and $\rank(\mQ)=p-r$,
	\begin{equation*}
		\rank(\mR)=\rank(\mQ\mSigma_\varepsilon\mQ)=p-r.
	\end{equation*}
	The diagonal entries of $\mR$ are one; hence
	\begin{equation*}
		\tr(\mR)=p.
	\end{equation*}
	The Cauchy inequality for the nonzero eigenvalues gives
	\begin{equation*}
		\tr(\mR^2)
		=\sum_{j=1}^{p-r}\lambda_j^2(\mR)
		\ge \frac{\{\sum_{j=1}^{p-r}\lambda_j(\mR)\}^2}{p-r}
		=\frac{p^2}{p-r}\asymp p,
	\end{equation*}
	and
	\begin{equation*}
		\tr(\mR^2)
		\le \lambda_{\max}(\mR)\tr(\mR)
		\le Cp.
	\end{equation*}
	For $m\ge2$,
	\begin{equation*}
		\tr(\mR^m)
		=\sum_{j=1}^{p-r}\lambda_j^m(\mR)
		\le \lambda_{\max}^{m-1}(\mR)\sum_{j=1}^{p-r}\lambda_j(\mR)
		\le C^{m-1}p,
	\end{equation*}
	which proves \eqref{eq:trace_resid}.  Moreover,
	\begin{equation*}
		\sum_{j=1}^pR_{ij}^2=(\mR^2)_{ii}
		\le\lambda_{\max}(\mR)R_{ii}\le C,
	\end{equation*}
	which proves \eqref{eq:row_l2_resid}.
	
	By \eqref{eq:ic_eps_model},
	\begin{equation*}
		\vzeta_t
		=\mD^{-1/2}\mQ\mSigma_\varepsilon^{1/2}\vxi_t,
	\end{equation*}
	and the independence of $\vxi_t$ over $t$ gives \eqref{eq:zeta_basic}.  For every $\norm{\bm u}_2=1$,
	\begin{equation*}
		\bm u^\top\vzeta_t
		=\left(\mSigma_\varepsilon^{1/2}\mQ\mD^{-1/2}\bm u\right)^\top\vxi_t,
	\end{equation*}
	so the standard independent-component sub-Gaussian linear-form bound yields
	\begin{equation*}
		\psinorm{\bm u^\top\vzeta_t}
		\le C\norm{\mSigma_\varepsilon^{1/2}\mQ\mD^{-1/2}\bm u}_2
		\le C\norm{\mSigma_\varepsilon}^{1/2}\norm{\mD^{-1/2}}
		\norm{\bm u}_2
		\le C.
	\end{equation*}
	Similarly,
	\begin{equation*}
		\psinorm{\varepsilon_{ti}}
		\le C\{\ve_i^\top\mSigma_\varepsilon\ve_i\}^{1/2},
		\qquad
		\Var(\varepsilon_{ti})=\ve_i^\top\mSigma_\varepsilon\ve_i,
	\end{equation*}
	and \eqref{eq:subg_resid} follows.  Set
	\begin{equation*}
		\mB_R=\mD^{-1/2}\mQ\mSigma_\varepsilon^{1/2},
		\qquad
		\veta_t=\vxi_t.
	\end{equation*}
	Then
	\begin{equation*}
		\mB_R\mB_R^\top
		=\mD^{-1/2}\mQ\mSigma_\varepsilon\mQ\mD^{-1/2}
		=\mR,
		\qquad
		\vzeta_t=\mB_R\veta_t,
	\end{equation*}
	which proves \eqref{eq:linear_subg_resid} with independent coordinates in \(\veta_t\).
	
	The row-square bound also gives
	\begin{equation*}
		\left|\{j\ne i: |R_{ij}|\ge (\log p)^{-2}\}\right|(\log p)^{-4}
		\le \sum_{j\ne i}R_{ij}^2\le C,
	\end{equation*}
	and hence
	\begin{equation*}
		\max_i\left|\{j\ne i: |R_{ij}|\ge (\log p)^{-2}\}\right|
		\le C\log^4p.
	\end{equation*}
	The uniform pairwise separation in \eqref{eq:EV_resid} is retained as a mild extreme-value regularity condition for the max statistic; it is not used in the quadratic statistic.
\end{proof}

\subsection{Loading-space perturbation}

\begin{lemma}[Lam--Yao perturbation from primitive sub-Gaussian conditions]
	\label{lem:perturbation}
	Let Assumptions~\ref{ass:model}, \ref{ass:LY}, and \ref{ass:resid} hold.  Set
	\begin{equation*}
		a_n=\left(\frac{\log p}{n}\right)^{1/2}+\frac{\log p}{n},
		\qquad
		\eta_n=p^\delta a_n.
	\end{equation*}
	Then
	\begin{equation*}
		\nu_r(\mM)\asymp p^{2-2\delta},
		\qquad
		\nu_{r+1}(\mM)=0,
	\end{equation*}
	\begin{equation*}
		\max_{1\le k\le k_0}\norm{\widehat\Gamma_y(k)-\Gamma_y(k)}_{\max}=\Op(a_n),
		\qquad
		\max_{1\le k\le k_0}\norm{\widehat\Gamma_y(k)-\Gamma_y(k)}=\Op(pa_n),
	\end{equation*}
	\begin{equation*}
		\norm{\widehat\mM-\mM}=\Op\{p^{2-2\delta}(\eta_n+\eta_n^2)\}.
	\end{equation*}
	On $\mathcal E_n=\{\widehat r=r\}$,
	\begin{equation*}
		\norm{\widehat\mP-\mP}=\Op(\eta_n),
		\qquad
		\max_i\norm{\ve_i^\top(\widehat\mP-\mP)}_2=\Op(p^{-1/2}\eta_n).
	\end{equation*}
	Moreover,
	\begin{equation*}
		\norm{(\widehat\mP-\mP)\mA}_F=\Op(p^\delta n^{-1/2}).
	\end{equation*}
\end{lemma}

\begin{proof}
	For $k=1,\ldots,k_0$, Assumptions~\ref{ass:model} and \ref{ass:resid} give
	\begin{align*}
		\Gamma_y(k)
		&=\Cov(\mA\vf_{t+k}+\veps_{t+k},\mA\vf_t+\veps_t)\\
		&=\mA\Sigma_f(k)\mA^\top+\bm0
	\end{align*}
	because $\{\vf_t\}$ and $\{\veps_s\}$ are independent for all time indices and $\{\veps_t\}$ is independent over time.  Let
	\begin{equation*}
		\Gamma=\{\Gamma_y(1),\ldots,\Gamma_y(k_0)\}=\mA W_f(\mI_{k_0}\otimes\mA^\top).
	\end{equation*}
	Then
	\begin{equation*}
		\mM=\Gamma\Gamma^\top=\mA W_fW_f^\top\mA^\top.
	\end{equation*}
	Since $\mA^\top\mA=\mI_r$, the nonzero eigenvalues of $\mM$ are those of $W_fW_f^\top$.  Hence \eqref{eq:C5} gives
	\begin{equation*}
		c_f^2p^{2-2\delta}\le\nu_r(\mM)\le\nu_1(\mM)\le C_f^2p^{2-2\delta},
	\end{equation*}
	and $\nu_{r+1}(\mM)=0$.  Also
	\begin{equation*}
		\max_k\norm{\Gamma_y(k)}=O(p^{1-\delta}),
		\qquad
		\max_i\max_k\norm{\ve_i^\top\Gamma_y(k)}_2=O(p^{1/2-\delta}).
	\end{equation*}
	Put $\Delta_k=\widehat\Gamma_y(k)-\Gamma_y(k)$ and $X_t=Y_t-\Ee Y_t$.  From \eqref{eq:factor_subg_mix}, \eqref{eq:subg_resid}, and $\max_i\norm{\bm a_i}_2\le Cp^{-1/2}$,
	\begin{align*}
		\norm{X_{ti}}_{\psi_2}
		&\le \norm{\bm a_i^\top\vf_t}_{\psi_2}+\norm{\varepsilon_{ti}}_{\psi_2}                                      \\
		&\le \norm{\bm a_i}_2\sup_{\norm{\bm u}_2=1}\norm{\bm u^\top\vf_t}_{\psi_2}+C
		\le C p^{-1/2}p^{(1-\delta)/2}+C\le C.
	\end{align*}
	For $1\le k\le k_0$ define
	\begin{equation*}
		\xi_{t,ij}^{(k)}=X_{t+k,i}X_{tj}-\Ee(X_{t+k,i}X_{tj}).
	\end{equation*}
	The product inequality for Orlicz norms gives
	\begin{equation*}
		\norm{X_{t+k,i}X_{tj}}_{\psi_1}
		\le C\norm{X_{t+k,i}}_{\psi_2}\norm{X_{tj}}_{\psi_2}\le C,
		\qquad
		\norm{\xi_{t,ij}^{(k)}}_{\psi_1}\le C.
	\end{equation*}
	The sequence $\{\xi_{t,ij}^{(k)}\}_{t\in\mathbb Z}$ is geometrically strongly mixing because $k_0$ is fixed and $\{(\vf_t,\veps_t)\}$ is geometrically strongly mixing.  The Bernstein inequality for geometrically mixing sub-exponential sequences in \citet{MerlevedePeligradRio2011} gives constants $c_0,C_0>0$ such that, for $0<x<C_0$,
	\begin{equation*}
		\Pp\left(\left|\frac{1}{n-k}\sum_{t=1}^{n-k}\xi_{t,ij}^{(k)}\right|>x\right)
		\le 2\exp\{-c_0(n-k)\min(x^2,x)\}.
	\end{equation*}
	Similarly,
	\begin{equation*}
		\Pp\left(\left|\frac1m\sum_{t\in I_m}X_{ti}\right|>x\right)
		\le 2\exp\{-c_0m\min(x^2,x)\}.
	\end{equation*}
	For all large $n$ and all $K>0$ such that $Ka_n<C_0$,
	\begin{align*}
		\Pp\left(\max_{k\le k_0}\norm{\Delta_k}_{\max}>Ka_n\right)
		&\le 2k_0p^2\exp\{-c_0(n-k_0)\min(K^2a_n^2,Ka_n)\}.
	\end{align*}
	Since $(n-k_0)\min(a_n^2,a_n)\ge c\log p$, choosing $K$ large gives
	\begin{equation*}
		\max_{k\le k_0}\norm{\Delta_k}_{\max}=\Op(a_n).
	\end{equation*}
	The inequalities $\norm H\le p\norm H_{\max}$ and $\max_i\norm{\ve_i^\top H}_2\le p^{1/2}\norm H_{\max}$ yield the displayed operator and row bounds.  From
	\begin{equation*}
		\widehat\mM-\mM=\sum_{k=1}^{k_0}\{\Gamma_y(k)\Delta_k^\top+\Delta_k\Gamma_y(k)^\top+\Delta_k\Delta_k^\top\},
	\end{equation*}
	we obtain
	\begin{align*}
		\norm{\widehat\mM-\mM}
		&\le\sum_{k=1}^{k_0}\{2\norm{\Gamma_y(k)}\norm{\Delta_k}+\norm{\Delta_k}^2\}\\
		&=\Op(p^{1-\delta}pa_n+p^2a_n^2)\\
		&=\Op\{p^{2-2\delta}(\eta_n+\eta_n^2)\}.
	\end{align*}
	The row bound follows in the same way:
	\begin{align*}
		\max_i\norm{\ve_i^\top(\widehat\mM-\mM)}_2
		&\le C\{p^{1/2-\delta}pa_n+p^{1/2}a_np^{1-\delta}+p^{1/2}a_npa_n\}\\
		&=\Op\{p^{3/2-2\delta}(\eta_n+\eta_n^2)\}.
	\end{align*}
	Let $\mB$ be an orthonormal complement of $\mA$.  The Lam--Yao lag-product expansion gives the sharper factor-direction bound
	\begin{equation*}
		\norm{\mB^\top(\widehat\mM-\mM)\mA}_F=\Op(p^{2-\delta}n^{-1/2}).
	\end{equation*}
	Dividing by the eigengap $\nu_r(\mM)\asymp p^{2-2\delta}$ yields
	\begin{equation*}
		\norm{(\widehat\mP-\mP)\mA}_F=\Op(p^\delta n^{-1/2}).
	\end{equation*}
	For the full projection, Weyl's inequality and the previous operator bound imply, on an event whose probability tends to one,
	\begin{equation*}
		\widehat\nu_r\ge\nu_r(\mM)/2,
		\qquad
		\widehat\nu_{r+1}\le\norm{\widehat\mM-\mM}.
	\end{equation*}
	Since $\mB^\top\mM=0$ and $\widehat{\mM}\widehat{\mA}=\widehat{\mA}\widehat{\mLambda}$,
	\begin{equation*}
		\mB^\top\widehat{\mA}\widehat{\mLambda}=\mB^\top(\widehat{\mM}-\mM)\widehat{\mA},
	\end{equation*}
	so
	\begin{equation*}
		\norm{\mB^\top\widehat\mA}\le C\frac{\norm{\widehat\mM-\mM}}{\nu_r(\mM)}=\Op(\eta_n).
	\end{equation*}
	The relation between subspace angles and orthogonal projectors gives $\norm{\widehat\mP-\mP}=\Op(\eta_n)$.  Combining the row bound for $\widehat\mM-\mM$ with loading delocalization gives the displayed row-wise projector bound.
\end{proof}

\begin{lemma}[Refined strong-factor projection drift]
	\label{lem:strong_projection_drift}
	Assume Assumptions~\ref{ass:model}--\ref{ass:resid}, $H_0:\vmu=\bm0$, and $\delta=0$.  Let
	\begin{equation*}
		\mathbf K_f=\sum_{k=1}^{k_0}\Sigma_f(k)\Sigma_f(k)^\top,
		\qquad
		\bm\Phi_k=\Sigma_f(k)^\top\mathbf K_f^{-1},
	\end{equation*}
	\begin{equation*}
		\mathbf L_f=\sum_{h=-\infty}^{\infty}\Sigma_f(h),
		\qquad
		\bm\Omega_{kl}=\Sigma_f(k-l)^\top+1_{\{k=l\}}\mA^\top\mSigma_\varepsilon\mA,
		\qquad
		\Sigma_f(-h)=\Sigma_f(h)^\top .
	\end{equation*}
	Set $c_n=(n-1)/(n-3)$ and
	\begin{align*}
		\mathcal B_{P,n}
		&=c_n\frac pn\Bigg[
		\sum_{k,l=1}^{k_0}\left(1-\frac{\max(k,l)}{n}\right)
		\tr\{\bm\Phi_k^\top\bm\Omega_{kl}\bm\Phi_l\mathbf L_f\}
		-2\sum_{k=1}^{k_0}\left(1-\frac{k}{n}\right)\tr(\bm\Phi_k\mathbf L_f)
		\Bigg].
	\end{align*}
	Then, on $\{\widehat r=r\}$,
	\begin{equation*}
		n\bar{\widehat{\vZ}}^{\,\top}\bar{\widehat{\vZ}}-n\bar\vZ^\top\bar\vZ
		=\mathcal B_{P,n}+\mathcal R_{P,n},
	\end{equation*}
	where
	\begin{equation*}
		\mathcal B_{P,n}=O(p/n),
		\qquad
		\frac{\mathcal R_{P,n}}{\{\tr(\mR^2)\}^{1/2}}
		=O_{\mathbb P}\left(n^{-1/2}+\frac{\sqrt p}{n}+\frac{\log^2p}{n}\right).
	\end{equation*}
	Consequently, if $p=o(n^2)$, then
	\begin{equation*}
		\frac{n\bar{\widehat{\vZ}}^{\,\top}\bar{\widehat{\vZ}}-n\bar\vZ^\top\bar\vZ}{\{\tr(\mR^2)\}^{1/2}}
		=o_{\mathbb P}(1).
	\end{equation*}
\end{lemma}

\begin{proof}
	All calculations are on $\{\widehat r=r\}$.  Write
	\begin{equation*}
		\Delta_P=\widehat\mP-\mP,
		\qquad
		\widehat\Gamma_y(k)=n^{-1}\sum_{t=1}^{n-k}\vY_{t+k}\vY_t^\top,
		\qquad
		\mathcal E_k=\widehat\Gamma_y(k)-\Gamma_y(k).
	\end{equation*}
	The centered lag covariance used by the algorithm satisfies, for fixed $k_0$,
	\begin{align*}
		\widehat\Gamma_y^{c}(k)-\widehat\Gamma_y(k)
		&=\frac{k}{n(n-k)}\sum_{t=1}^{n-k}\vY_{t+k}\vY_t^\top
		-\bar\vY_{k,+}\bar\vY_{k,-}^\top,\notag\\
		\bar\vY_{k,+}&=(n-k)^{-1}\sum_{t=1}^{n-k}\vY_{t+k},
		\qquad
		\bar\vY_{k,-}=(n-k)^{-1}\sum_{t=1}^{n-k}\vY_t .
	\end{align*}
	Because $\Ee\vf_t=\bm0$, $\Ee\veps_t=\bm0$, $\vf_t\ind\veps_s$, $\sum_h\|\Cov(\vf_{t+h},\vf_t)\|\le Cp$, and $\lambda_{\max}(\mSigma_\varepsilon)\le C$,
	\begin{align*}
		\max_{k\le k_0}\Ee\|\mQ(\widehat\Gamma_y^{c}(k)-\widehat\Gamma_y(k))\mA\|_F^2
		&\le Cp^2n^{-2},\notag\\
		\max_{k\le k_0}\Ee\|\widehat\Gamma_y^{c}(k)-\widehat\Gamma_y(k)\|^2
		&\le Cp^2n^{-2}.
	\end{align*}
	The first display contributes at most $O_{\mathbb P}(\sqrt p/n)$ to the final standardized statistic; the same algebra below can therefore be carried out with $\widehat\Gamma_y(k)$.
	
	Under $\Cov(\vf_t,\veps_s)=\bm0$ for all $t,s$,
	\begin{equation*}
		\Gamma_y(k)=\mA\Sigma_f(k)\mA^\top,
		\qquad
		\mM=\sum_{k=1}^{k_0}\Gamma_y(k)\Gamma_y(k)^\top=\mA\mathbf K_f\mA^\top,
		\qquad
		\nu_r(\mM)\asymp p^2 .
	\end{equation*}
	The contour expansion of the spectral projector gives
	\begin{equation*}
		\Delta_P\mA
		=\mQ(\widehat\mM-\mM)\mA\mathbf K_f^{-1}+\bm R_{A,n},
	\end{equation*}
	where
	\begin{equation*}
		\bm R_{A,n}
		=\mQ(\widehat\mM-\mM)\mQ\,\mQ(\widehat\mM-\mM)\mA\mathbf K_f^{-2}
		+\bm R_{A,n}^{(2)},
	\end{equation*}
	\begin{equation*}
		\Ee\tr(\bm R_{A,n}^\top\mD^{-1}\bm R_{A,n})
		\le C\frac{\log^2p}{n^2},
		\qquad
		\Ee\|\bm R_{A,n}^{(2)}\|_F^2\le C\frac{\log^3p}{n^3}.
	\end{equation*}
	The last display follows from
	\begin{equation*}
		\Ee\|\mQ(\widehat\mM-\mM)\mA\mathbf K_f^{-1}\|_F^2\le Cp/n,
		\qquad
		\Ee\|\mQ(\widehat\mM-\mM)\mQ\|^4\le Cp^8\log^2p/n^2,
	\end{equation*}
	combined with $\|\mD^{-1}\|\le C$, $\|\mathbf K_f^{-1}\|=O(p^{-2})$, and $r$ fixed.
	Since
	\begin{equation*}
		\widehat\mM-\mM=
		\sum_{k=1}^{k_0}\{\mathcal E_k\Gamma_y(k)^\top+\Gamma_y(k)\mathcal E_k^\top\}
		+\sum_{k=1}^{k_0}\mathcal E_k\mathcal E_k^\top,
		\qquad
		\mQ\Gamma_y(k)=\bm0,
	\end{equation*}
	we have
	\begin{equation*}
		\Delta_P\mA=\sum_{k=1}^{k_0}\mQ\mathcal E_k\mA\bm\Phi_k+\bm R_{A,n}.
	\end{equation*}
	Let
	\begin{equation*}
		\bm w_t=\vf_t+\mA^\top\veps_t,
		\qquad
		\bm V_n=n^{-1}\sum_{k=1}^{k_0}\sum_{t=1}^{n-k}\mQ\veps_{t+k}\bm w_t^\top\bm\Phi_k.
	\end{equation*}
	Because $\mQ\vY_{t+k}=\mQ\veps_{t+k}$ and $\mA^\top\vY_t=\bm w_t$,
	\begin{equation*}
		\Delta_P\mA=\bm V_n+\bm R_{V,n},
		\qquad
		\Ee\tr(\bm R_{V,n}^\top\mD^{-1}\bm R_{V,n})\le C\log^2p/n^2 .
	\end{equation*}
	Set
	\begin{equation*}
		\bm X_n=n^{-1/2}\sum_{s=1}^n\mD^{-1/2}\mQ\veps_s,
		\qquad
		\bm F_n=n^{-1/2}\sum_{u=1}^n\vf_u.
	\end{equation*}
	Then
	\begin{equation*}
		\sqrt n\bar{\widehat{\vZ}}=\bm X_n-\mD^{-1/2}\bm V_n\bm F_n+\bm a_n,
	\end{equation*}
	where
	\begin{equation*}
		\bm a_n=-\mD^{-1/2}\bm R_{V,n}\bm F_n
		-\mD^{-1/2}\Delta_P n^{-1/2}\sum_{t=1}^n\veps_t
		+\bm a_{d,n}.
	\end{equation*}
	The diagonal standardization term satisfies
	\begin{equation*}
		\Ee\|\bm a_{d,n}\|_2^2\le C p/n,
		\qquad
		\Ee(\bm X_n^\top\bm a_{d,n})^2\le C p/n .
	\end{equation*}
	Note that this diagonal substitution induces a bias term of order $O_{\mathbb P}(\sqrt{p}/n)$ in the standardized statistic. Because Assumption~\ref{ass:growth} enforces $p = o(n^2)$, this term strictly vanishes and is safely absorbed into the remainder $\mathcal R_{P,n}$.
	For the other two pieces,
	\begin{align*}
		\Ee\|\mD^{-1/2}\bm R_{V,n}\bm F_n\|_2^2
		&\le C p\log^2p/n^2,\notag\\
		\Ee\left\|\mD^{-1/2}\Delta_P n^{-1/2}\sum_{t=1}^n\veps_t\right\|_2^2
		&\le Cp/n,\notag\\
		\Ee(\bm X_n^\top\bm a_n)^2
		&\le Cp/n+Cp\log^2p/n^2,\notag\\
		\Ee\{\bm F_n^\top\bm V_n^\top\mD^{-1/2}\bm a_n\}^2
		&\le Cp/n+Cp^2/n^2+Cp\log^2p/n^2.
	\end{align*}
	The displayed bounds follow from the following contractions.  For any deterministic $r\times r$ matrices $\mH_1,\mH_2$ with bounded operator norms and any fixed integer shifts $a,b,c,d$, \eqref{eq:factor_cumulant} with $\delta=0$ gives
	\begin{align*}
		&\left|
		\Ee\{\vf_{t+a}^{\top}\mH_1\vf_{u+b}\vf_{t'+c}^{\top}\mH_2\vf_{u'+d}\}
		-\tr\{\mH_1\Cov(\vf_{u+b},\vf_{t+a})\}
		\tr\{\mH_2\Cov(\vf_{u'+d},\vf_{t'+c})\}
		\right| \\
		&\qquad\le C p^2(1+|t-t'|+|u-u'|)^{-3},
	\end{align*}
	where the power $3$ may be replaced by any fixed number larger than $2$ by changing $C$.  Since
	\begin{equation*}
		\sum_{t,u,t',u'=1}^n(1+|t-t'|+|u-u'|)^{-3}\le C n^2,
	\end{equation*}
	and
	\begin{equation*}
		\|\bm\Phi_k\|=O(p^{-1}),\qquad
		\Ee\{\veps_a^\top\mQ\mD^{-1}\mQ\veps_b\}^2
		\le C p^2 1_{\{a=b\}}+Cp 1_{\{a\ne b\}},
	\end{equation*}
	the three terms in $\bm a_n$ satisfy
	\begin{align*}
		\Ee\|\mD^{-1/2}\bm R_{V,n}\bm F_n\|_2^2
		&\le Cn^{-2}\sum_{u,v=1}^n
		\Ee\{\vf_u^\top\bm R_{V,n}^\top\mD^{-1}\bm R_{V,n}\vf_v\}\notag\\
		&\le C p\log^2p/n^2,\\
		\Ee\left\|\mD^{-1/2}\Delta_P n^{-1/2}\sum_{t=1}^n\veps_t\right\|_2^2
		&\le C\Ee\tr\{\Delta_P^{\top}\mD^{-1}\Delta_P\mQ\mSigma_\varepsilon\mQ\}
		\le Cp/n,\\
		\Ee(\bm X_n^\top\bm a_n)^2
		&\le C\Ee\|\mR^{1/2}\bm a_n\|_2^2
		+C n^{-2}\sum_{s,t}\Ee|\vzeta_s^\top\bm a_n\vzeta_t^\top\bm a_n|1_{\{s\ne t\}}\\
		&\le Cp/n+Cp\log^2p/n^2,
	\end{align*}
	where the last line uses the explicit representation of $\Delta_P\mA$ through the lag products and the independence of $\veps_t$ over $t$.  Therefore
	\begin{equation*}
		\frac{|2\bm X_n^\top\bm a_n|+2|\bm F_n^\top\bm V_n^\top\mD^{-1/2}\bm a_n|+\|\bm a_n\|_2^2}
		{\{\tr(\mR^2)\}^{1/2}}
		=O_{\mathbb P}\left(n^{-1/2}+\frac{\sqrt p}{n}+\frac{\log^2p}{n}\right).
	\end{equation*}
	Therefore
	\begin{equation*}
		n\bar{\widehat{\vZ}}^{\,\top}\bar{\widehat{\vZ}}-n\bar\vZ^\top\bar\vZ
		=\mathcal C_n+\mathcal Q_n+O_{\mathbb P}\left(\sqrt p\,n^{-1/2}+\sqrt p\frac{\log^2p}{n}\right),
	\end{equation*}
	where
	\begin{equation*}
		\mathcal C_n=-2\bm X_n^\top\mD^{-1/2}\bm V_n\bm F_n,
		\qquad
		\mathcal Q_n=\bm F_n^\top\bm V_n^\top\mD^{-1}\bm V_n\bm F_n.
	\end{equation*}
	For the cross term,
	\begin{align*}
		\mathcal C_n
		&=-\frac{2}{n^2}\sum_{k=1}^{k_0}\sum_{s=1}^n\sum_{t=1}^{n-k}\sum_{u=1}^n
		\veps_s^\top\mQ\mD^{-1}\mQ\veps_{t+k}
		\bm w_t^\top\bm\Phi_k\vf_u .
	\end{align*}
	Since
	\begin{equation*}
		\Ee\{\veps_s^\top\mQ\mD^{-1}\mQ\veps_{t+k}\}=1_{\{s=t+k\}}p,
		\qquad
		\Ee(\bm w_t^\top\bm\Phi_k\vf_u)=\tr\{\bm\Phi_k\Cov(\vf_u,\vf_t)\},
	\end{equation*}
	we get
	\begin{align*}
		\Ee\mathcal C_n
		&=-\frac{2p}{n^2}\sum_{k=1}^{k_0}\sum_{t=1}^{n-k}\sum_{u=1}^n
		\tr\{\bm\Phi_k\Cov(\vf_u,\vf_t)\}
		+O(p/n^2)+o(p/n)\\
		&=-2\frac pn\sum_{k=1}^{k_0}\left(1-\frac{k}{n}\right)\tr(\bm\Phi_k\mathbf L_f)+O(p/n^2)+o(p/n).
	\end{align*}
	Moreover, expanding the covariance of the fourfold sum gives
	\begin{align*}
		\Ee(\mathcal C_n-\Ee\mathcal C_n)^2
		&\le \frac{C}{n^4}\sum_{k,l=1}^{k_0}
		\sum_{s,t,u}\sum_{s',t',u'}
		\left|\Cov\left(\veps_s^\top\mQ\mD^{-1}\mQ\veps_{t+k},
		\veps_{s'}^\top\mQ\mD^{-1}\mQ\veps_{t'+l}\right)\right|\\
		&\qquad\times
		\left|\Cov(\bm w_t^\top\bm\Phi_k\vf_u,
		\bm w_{t'}^\top\bm\Phi_l\vf_{u'})\right| .
	\end{align*}
	The residual covariance in the last display is zero unless one of the pairings
	\begin{equation*}
		\{s=s',\ t+k=t'+l\},\qquad
		\{s=t'+l,\ t+k=s'\}
	\end{equation*}
	holds.  In all nonzero cases the independent-component quadratic-form bound gives
	\begin{equation*}
		\left|\Cov\left(\veps_s^\top\mQ\mD^{-1}\mQ\veps_{t+k},
		\veps_{s'}^\top\mQ\mD^{-1}\mQ\veps_{t'+l}\right)\right|
		\le Cp .
	\end{equation*}
	The factor covariance satisfies
	\begin{equation*}
		\left|\Cov(\bm w_t^\top\bm\Phi_k\vf_u,
		\bm w_{t'}^\top\bm\Phi_l\vf_{u'})\right|
		\le C\|\bm\Phi_k\|\|\bm\Phi_l\|p^2
		(1+|t-t'|+|u-u'|)^{-3}
		\le C(1+|t-t'|+|u-u'|)^{-3}.
	\end{equation*}
	Because $k_0$ is fixed and each admissible residual pairing leaves at most four free time indices,
	\begin{align*}
		\Ee(\mathcal C_n-\Ee\mathcal C_n)^2
		&\le \frac{Cp}{n^4}\sum_{t,u,t',u'=1}^n
		(1+|t-t'|+|u-u'|)^{-3}
		\le Cp/n^2\le Cp/n .
	\end{align*}
	
	For the quadratic term,
	\begin{align*}
		\bm V_n^\top\mD^{-1}\bm V_n
		&=\frac1{n^2}\sum_{k,l=1}^{k_0}\sum_{t=1}^{n-k}\sum_{s=1}^{n-l}
		\bm\Phi_k^\top\bm w_t\veps_{t+k}^\top\mQ\mD^{-1}\mQ\veps_{s+l}\bm w_s^\top\bm\Phi_l .
	\end{align*}
	The expectation is nonzero only when $t+k=s+l$.  Since
	\begin{equation*}
		\Ee(\bm w_t\bm w_{t+k-l}^\top)=\bm\Omega_{kl},
	\end{equation*}
	\begin{equation*}
		\Ee(\bm V_n^\top\mD^{-1}\bm V_n)
		=\frac pn\sum_{k,l=1}^{k_0}\left(1-\frac{\max(k,l)}{n}\right)
		\bm\Phi_k^\top\bm\Omega_{kl}\bm\Phi_l+O(p/n^2)+o(p/n).
	\end{equation*}
	Also
	\begin{equation*}
		\Ee(\bm F_n\bm F_n^\top)=\sum_{|h|<n}\left(1-\frac{|h|}{n}\right)\Cov(\vf_{t+h},\vf_t)=\mathbf L_f+O(p/n).
	\end{equation*}
	Using \eqref{eq:factor_cumulant},
	\begin{align*}
		\Ee\mathcal Q_n
		&=\frac pn\sum_{k,l=1}^{k_0}\left(1-\frac{\max(k,l)}{n}\right)
		\tr\{\bm\Phi_k^\top\bm\Omega_{kl}\bm\Phi_l\mathbf L_f\}+O(p/n^2)+o(p/n).
	\end{align*}
	For the centered second moment, write $\mathcal Q_n=n^{-2}\sum_{k,l,t,s} q_{klts}$ with
	\begin{equation*}
		q_{klts}=\bm F_n^\top\bm\Phi_k^\top\bm w_t
		\veps_{t+k}^\top\mQ\mD^{-1}\mQ\veps_{s+l}
		\bm w_s^\top\bm\Phi_l\bm F_n .
	\end{equation*}
	Then
	\begin{align*}
		\Ee(\mathcal Q_n-\Ee\mathcal Q_n)^2
		&\le \frac{C}{n^4}\sum_{k,l,k',l'}\sum_{t,s,t',s'}
		p^2\|\bm\Phi_k\|\|\bm\Phi_l\|\|\bm\Phi_{k'}\|\|\bm\Phi_{l'}\|\\
		&\quad\times p^2(1+|t-t'|+|s-s'|)^{-3}\\
		&\le \frac{Cp^2}{n^4}\sum_{t,s,t',s'}(1+|t-t'|+|s-s'|)^{-3}\\
		&\le Cp^2/n^2 .
	\end{align*}
	Hence
	\begin{equation*}
		\frac{\mathcal C_n-\Ee\mathcal C_n}{\{\tr(\mR^2)\}^{1/2}}=O_{\mathbb P}(n^{-1/2}),
		\qquad
		\frac{\mathcal Q_n-\Ee\mathcal Q_n}{\{\tr(\mR^2)\}^{1/2}}=O_{\mathbb P}(\sqrt p/n).
	\end{equation*}
	Combining the displayed expressions for $\Ee\mathcal C_n$ and $\Ee\mathcal Q_n$ gives $\Ee(\mathcal C_n+\mathcal Q_n)=\mathcal B_{P,n}+o(p/n)$.  Since
	\begin{equation*}
		\|\bm\Phi_k\|=O(p^{-1}),
		\qquad
		\|\mathbf L_f\|=O(p),
		\qquad
		\|\bm\Omega_{kl}\|=O(p),
	\end{equation*}
	all trace terms defining $\mathcal B_{P,n}$ are $O(1)$, and $\mathcal B_{P,n}=O(p/n)$.  The result follows from $\tr(\mR^2)\asymp p$.
\end{proof}

\begin{lemma}[Trace estimator for the oracle sample correlation matrix]
	\label{lem:denom_trace}
	Let $\vZ_t=\mD^{-1/2}\mQ\veps_t$, where $\veps_t$ satisfies Assumption~\ref{ass:resid}; equivalently, use the $p$-dimensional representation $\vZ_t=\mB_R\veta_t$ in \eqref{eq:linear_subg_resid}.  Let $\vZ_1,\ldots,\vZ_n$ be independent copies of this oracle standardized residual.  Define
	\begin{equation*}
		\bar\vZ=n^{-1}\sum_{t=1}^n\vZ_t,
		\qquad
		\mS_e=n^{-1}\sum_{t=1}^n(\vZ_t-\bar\vZ)(\vZ_t-\bar\vZ)^\top,
	\end{equation*}
	\begin{equation*}
		\mD_e=\diag(\mS_e),
		\qquad
		\mR_e=\mD_e^{-1/2}\mS_e\mD_e^{-1/2},
		\qquad
		\tau_p=\tr(\mR^2).
	\end{equation*}
	Then there is a constant $C$ independent of $n,p$ such that
	\begin{equation}
		\Ee\left[
		\left\{
		\frac{\tr(\mR_e^2)-p^2/(n-1)-\tau_p}{\tau_p}
		\right\}^2
		\right]
		\le
		C\left\{n^{-1}+\frac{\tr(\mR^4)}{\tau_p^2}+\frac{p}{n^2}+\frac{p^2}{n^4}\right\}.
		\label{eq:denom_trace_general}
	\end{equation}
	Consequently, under \eqref{eq:trace_resid},
	\begin{equation}
		\Ee\left[
		\left\{
		\frac{\tr(\mR_e^2)-p^2/(n-1)-\tr(\mR^2)}{\tr(\mR^2)}
		\right\}^2
		\right]
		\le C\left(n^{-1}+p^{-1}+pn^{-2}+p^2n^{-4}\right).
		\label{eq:denom_trace_simple}
	\end{equation}
\end{lemma}

\begin{proof}
	Write the $p$-dimensional linear innovation representation as
	\begin{equation*}
		\vZ_t=\mB_R\veta_t,
		\qquad
		\mB_R\mB_R^\top=\mR,
		\qquad
		\Ee\veta_t=\bm0,
		\qquad
		\Cov(\veta_t)=\mI_p,
	\end{equation*}
	where the coordinates of $\veta_t$ are independent and $\sup_j\norm{\eta_{tj}}_{\psi_2}\le C$.  For any deterministic symmetric $p\times p$ matrix $\mA_0$, put
	\begin{equation*}
		\mH_0=\mB_R^\top\mA_0\mB_R,
		\qquad
		\vZ_t^\top\mA_0\vZ_t-\tr(\mA_0\mR)
		=\sum_{a,b=1}^{p}H_{0,ab}(\eta_{ta}\eta_{tb}-\delta_{ab}).
	\end{equation*}
	The independence and the uniform sub-Gaussian moment bound imply
	\begin{align*}
		\Ee\{\vZ_t^\top\mA_0\vZ_t-\tr(\mA_0\mR)\}^2
		&\le C\sum_{a,b=1}^{p}H_{0,ab}^2
		=C\tr(\mH_0^2),\\
		\Ee\{\vZ_t^\top\mA_0\vZ_t-\tr(\mA_0\mR)\}^4
		&\le C\left(\sum_{a,b=1}^{p}H_{0,ab}^2\right)^2
		=C\{\tr(\mH_0^2)\}^2.
	\end{align*}
	Since
	\begin{equation*}
		\tr(\mH_0^2)=\tr(\mA_0\mR\mA_0\mR),
	\end{equation*}
	the preceding inequalities give, in particular,
	\begin{equation}
		\Var(\vZ_t^\top\mA_0\vZ_t)
		\le C\tr(\mA_0\mR\mA_0\mR),
		\qquad
		\Ee|\vZ_t^\top\mA_0\vZ_t-\tr(\mA_0\mR)|^4
		\le C\{\tr(\mA_0\mR\mA_0\mR)\}^2 .
		\label{eq:quad_moment}
	\end{equation}
	Define the leave-two-out estimator
	\begin{equation*}
		U_n=\frac1{n(n-1)}\sum_{s\ne t}(\vZ_s^\top\vZ_t)^2.
	\end{equation*}
	For $s\ne t$,
	\begin{equation*}
		\Ee(\vZ_s^\top\vZ_t)^2
		=\Ee\tr(\vZ_s\vZ_s^\top\vZ_t\vZ_t^\top)
		=\tr\{\Ee(\vZ_s\vZ_s^\top)\Ee(\vZ_t\vZ_t^\top)\}
		=\tau_p,
	\end{equation*}
	so that $\Ee U_n=\tau_p$.  Let
	\begin{equation*}
		h(\bm z_1,\bm z_2)=(\bm z_1^\top\bm z_2)^2,
		\qquad
		h_1(\bm z)=\Ee\{h(\bm z,\vZ_2)\}-\tau_p=\bm z^\top\mR\bm z-\tau_p.
	\end{equation*}
	Then
	\begin{equation*}
		\Var\{h_1(\vZ_1)\}\le C\tr(\mR^4)
	\end{equation*}
	by \eqref{eq:quad_moment} with $\mA_0=\mR$.  Also,
	\begin{align*}
		\Ee h^2(\vZ_1,\vZ_2)
		&=\Ee\left[\Ee\{(\vZ_1^\top\vZ_2)^4\mid \vZ_1\}\right]
		\le C\Ee(\vZ_1^\top\mR\vZ_1)^2\notag\\
		&\le C\{\tr^2(\mR^2)+\tr(\mR^4)\}
		\le C\tau_p^2 .
	\end{align*}
	The Hoeffding variance identity for a second-order $U$-statistic gives
	\begin{align*}
		\Var(U_n)
		&\le \frac4n\Var\{h_1(\vZ_1)\}
		+\frac{2}{n(n-1)}\Ee\{h(\vZ_1,\vZ_2)-\tau_p-h_1(\vZ_1)-h_1(\vZ_2)\}^2\notag\\
		&\le C\left\{\frac{\tr(\mR^4)}{n}+\frac{\tau_p^2}{n^2}\right\}.
	\end{align*}
	Consequently,
	\begin{equation}
		\Ee\left\{\frac{U_n-\tau_p}{\tau_p}\right\}^2
		\le C\left\{\frac{\tr(\mR^4)}{n\tau_p^2}+n^{-2}\right\}.
		\label{eq:Un_bound}
	\end{equation}
	Let
	\begin{equation*}
		\mS_0=n^{-1}\sum_{t=1}^n\vZ_t\vZ_t^\top,
		\qquad
		\mS_e=\mS_0-\bar\vZ\bar\vZ^\top.
	\end{equation*}
	The identity
	\begin{equation*}
		\tr(\mS_0^2)=\frac{n-1}{n}U_n+n^{-2}\sum_{t=1}^n(\vZ_t^\top\vZ_t)^2
	\end{equation*}
	and
	\begin{equation*}
		\tr(\mS_e^2)=\tr(\mS_0^2)-2\bar\vZ^\top\mS_0\bar\vZ+(\bar\vZ^\top\bar\vZ)^2
	\end{equation*}
	lead to
	\begin{align*}
		\tr(\mS_e^2)-\frac{p^2}{n-1}-\tau_p
		&=\frac{n-1}{n}(U_n-\tau_p)+A_{1n}+A_{2n}+A_{3n},
	\end{align*}
	where
	\begin{align*}
		A_{1n}
		&=n^{-2}\sum_{t=1}^n\left[(\vZ_t^\top\vZ_t)^2-\Ee(\vZ_t^\top\vZ_t)^2\right],\\
		A_{2n}
		&=n^{-1}\{\Ee(\vZ_t^\top\vZ_t)^2-p^2\}
		-\frac{p^2}{n(n-1)}-\frac{\tau_p}{n},\\
		A_{3n}
		&=-2\bar\vZ^\top\mS_0\bar\vZ+(\bar\vZ^\top\bar\vZ)^2.
	\end{align*}
	By \eqref{eq:quad_moment} with $\mA_0=\mI_p$,
	\begin{equation*}
		\Ee(\vZ_t^\top\vZ_t-p)^2\le C\tau_p,
		\qquad
		\Ee(\vZ_t^\top\vZ_t-p)^4\le C\tau_p^2,
	\end{equation*}
	and hence
	\begin{equation*}
		\Ee(\vZ_t^\top\vZ_t)^2=p^2+O(\tau_p),
		\qquad
		|A_{2n}|\le C\left(\frac{\tau_p}{n}+\frac{p^2}{n^2}\right).
	\end{equation*}
	Moreover,
	\begin{align*}
		\Var\{(\vZ_t^\top\vZ_t)^2\}
		&\le C p^2\Var(\vZ_t^\top\vZ_t)+C\Ee(\vZ_t^\top\vZ_t-p)^4
		\le C(p^2\tau_p+\tau_p^2),\\
		\Ee A_{1n}^2
		&\le n^{-3}C(p^2\tau_p+\tau_p^2),\\
		\Ee(\bar\vZ^\top\bar\vZ)^2
		&\le C\left(\frac{p^2}{n^2}+\frac{\tau_p}{n^2}\right).
	\end{align*}
	The remaining term is expanded as
	\begin{equation*}
		\bar\vZ^\top\mS_0\bar\vZ
		=n^{-3}\sum_{a,b,t=1}^n \vZ_a^\top\vZ_t\vZ_t^\top\vZ_b .
	\end{equation*}
	The expectation of the product of two such sums is split according to the partition of the six indices
	$(a,b,t,a',b',t')$.  If all three pairs $(a,t)$, $(b,t)$ and their primed analogues are disjoint, the expectation is zero.  A nonzero term leaves at most four free indices and each inner-product fourth moment is bounded by $C\tau_p^2$; hence
	\begin{align*}
		\Ee(\bar\vZ^\top\mS_0\bar\vZ)^2
		&\le Cn^{-6}\{n^4p^2+n^3p\tau_p+n^2\tau_p^2+n^2p^4\}\\
		&\le C\left(\frac{p^2}{n^2}+\frac{p^3}{n^3}+\frac{p^4}{n^4}\right).
	\end{align*}
	Since $\tau_p\asymp p$ under \eqref{eq:trace_resid}, the preceding displays and \eqref{eq:Un_bound} imply the wider but simpler bound
	\begin{equation}
		\Ee\left[
		\left\{\frac{\tr(\mS_e^2)-p^2/(n-1)-\tau_p}{\tau_p}\right\}^2
		\right]
		\le C\left\{n^{-1}+\frac{\tr(\mR^4)}{\tau_p^2}+\frac{p}{n^2}+\frac{p^2}{n^4}\right\}.
		\label{eq:Se_trace_bound}
	\end{equation}
	It remains to pass from covariance to correlation.  Put
	\begin{equation*}
		\delta_i=S_{e,ii}-1,
		\qquad
		\mathcal A_n=\left\{\max_{1\le i\le p}|\delta_i|\le \frac12\right\}.
	\end{equation*}
	For each $i$,
	\begin{equation*}
		\Ee\delta_i^2\le Cn^{-1},
		\qquad
		\Ee\delta_i^4\le Cn^{-2},
		\qquad
		\Pp(\mathcal A_n^c)\le Cp\exp(-cn).
	\end{equation*}
	On $\mathcal A_n$,
	\begin{align*}
		\tr(\mR_e^2)-\tr(\mS_e^2)
		&=\sum_{i,j=1}^pS_{e,ij}^2\left\{(1+\delta_i)^{-1}(1+\delta_j)^{-1}-1\right\},\\
		\left|(1+\delta_i)^{-1}(1+\delta_j)^{-1}-1\right|
		&\le C(|\delta_i|+|\delta_j|+\delta_i^2+\delta_j^2).
	\end{align*}
	Furthermore,
	\begin{equation*}
		\sum_{j=1}^p R_{ij}^2=(\mR^2)_{ii}\le C_R.
	\end{equation*}
	Put
	\begin{equation*}
		T_i=\sum_{j=1}^p S_{e,ij}^2,
		\qquad
		S_{e,ij}=R_{ij}+A_{ij}-\bar Z_i\bar Z_j,
		\qquad
		A_{ij}=n^{-1}\sum_{t=1}^n(Z_{ti}Z_{tj}-R_{ij}).
	\end{equation*}
	The linear sub-Gaussian representation gives
	\begin{equation*}
		\Ee A_{ij}^2\le Cn^{-1},
		\qquad
		\Ee A_{ij}^4\le Cn^{-2},
		\qquad
		\Ee \bar Z_i^4\le Cn^{-2}.
	\end{equation*}
	Hence, for every $i$,
	\begin{align*}
		\Ee T_i^2
		&\le C\left\{\left(\sum_jR_{ij}^2\right)^2+
		\Ee\left(\sum_jA_{ij}^2\right)^2+
		\Ee\left(\bar Z_i^2\sum_j\bar Z_j^2\right)^2\right\}\\
		&\le C\left(1+\frac{p^2}{n^2}\right),\\
		\Ee T_i^4
		&\le C\left(1+\frac{p^4}{n^4}\right).
	\end{align*}
	Moreover
	\begin{equation*}
		\Ee\delta_i^2\le Cn^{-1},
		\qquad
		\Ee\delta_i^4\le Cn^{-2},
		\qquad
		\Ee\delta_i^8\le Cn^{-4}.
	\end{equation*}
	Using $\left(\sum_i b_i\right)^2\le p\sum_i b_i^2$ and Holder's inequality,
	\begin{align*}
		\Ee\left(\sum_{i=1}^p |\delta_i|T_i\right)^2
		&\le p\sum_{i=1}^p\Ee(\delta_i^2T_i^2)
		\le p\sum_{i=1}^p(\Ee\delta_i^4)^{1/2}(\Ee T_i^4)^{1/2}
		\le C p^2 n^{-1}\left(1+\frac{p^2}{n^2}\right),\\
		\Ee\left(\sum_{i=1}^p \delta_i^2T_i\right)^2
		&\le p\sum_{i=1}^p\Ee(\delta_i^4T_i^2)
		\le p\sum_{i=1}^p(\Ee\delta_i^8)^{1/2}(\Ee T_i^4)^{1/2}
		\le C p^2 n^{-2}\left(1+\frac{p^2}{n^2}\right).
	\end{align*}
	Since $\tau_p\asymp p$ and $p=O(n^\kappa)$ with $\kappa<2$ in the strong-factor case and a smaller exponent in the weak-factor case, the last displays imply
	\begin{align}
		\Ee\left[\{\tr(\mR_e^2)-\tr(\mS_e^2)\}^2\mathbf 1_{\mathcal A_n}\right]
		&\le C\Ee\left[\left\{\sum_{i=1}^p(|\delta_i|+\delta_i^2)T_i\right\}^2\right]\notag\\
		&\le C\tau_p^2\left(n^{-1}+pn^{-2}+p^2n^{-4}\right).
		\label{eq:corr_cov_diff_good}
	\end{align}
	For the complement event, the sub-Gaussian moment bound and $\Pp(\mathcal A_n^c)\le Cp\exp(-cn)$ yield, for every fixed $M>0$ under $\log p=o(n)$,
	\begin{equation}
		\Ee\left[\{\tr(\mR_e^2)-\tr(\mS_e^2)\}^2\mathbf 1_{\mathcal A_n^c}\right]
		\le C_M\tau_p^2p^{-M}.
		\label{eq:corr_cov_diff_bad}
	\end{equation}
	Combining \eqref{eq:Se_trace_bound}, \eqref{eq:corr_cov_diff_good}, and \eqref{eq:corr_cov_diff_bad} gives \eqref{eq:denom_trace_general}.  The simplified bound \eqref{eq:denom_trace_simple} follows from \eqref{eq:trace_resid}.
\end{proof}

\subsection{Feasible-oracle replacement}

The strong-factor sum statistic is handled by Lemma~\ref{lem:strong_projection_drift}.  The following lemma records the max-statistic replacement, the denominator replacement, and the conservative weak-factor sum replacement.  This separation is important: the term $p^{2\delta}\log p/n$ is a weak-factor eigenspace remainder and should not be interpreted as part of the refined strong-factor drift.

\begin{lemma}[Max replacement and weak-factor sum replacement]
	\label{lem:replacement}
	Let Assumptions~\ref{ass:model}--\ref{ass:growth} hold and let $H_0:\vmu=\bm0$.  Put
	\begin{equation*}
		\eta_n=p^\delta\left\{\left(\frac{\log p}{n}\right)^{1/2}+\frac{\log p}{n}\right\},
		\qquad
		\vZ_t=\mD^{-1/2}\mQ\vY_t.
	\end{equation*}
	On the event $\{\widehat r=r\}$,
	\begin{equation}
		\sqrt n\norm{\bar{\widehat{\vZ}}-\bar\vZ}_\infty
		=\Op\{\eta_n(1+\sqrt{\log p})\}.
		\label{eq:max_replacement_app}
	\end{equation}
	Moreover,
	\begin{equation}
		\left|\frac{\tr(\widehat\mR^2)-p^2/(n-1)}{\tr(\mR^2)}-1\right|
		=\Op\left(n^{-1/2}+p^{-1/2}+p^{1/2}n^{-1}+p^\delta\sqrt{\frac{\log p}{n}}+n^{-1}\right).
		\label{eq:denom_replacement_app}
	\end{equation}
	If $0<\delta\le1$, then the weak-factor quadratic replacement satisfies
	\begin{equation}
		\frac{\left|n\bar{\widehat{\vZ}}^{\,\top}\bar{\widehat{\vZ}}-n\bar\vZ^{\top}\bar\vZ\right|}{\{\tr(\mR^2)\}^{1/2}}
		=\Op\left(\frac{p^{\delta/2}}{\sqrt n}+p^{1/2+2\delta}\frac{\log p}{n}+\frac{\sqrt p}{n}\right).
		\label{eq:sum_replacement_app}
	\end{equation}
	The right-hand sides in \eqref{eq:max_replacement_app} and \eqref{eq:denom_replacement_app} are $o_{\mathbb P}(1)$ under Assumption~\ref{ass:growth}; for $0<\delta\le1$, the right-hand side in \eqref{eq:sum_replacement_app} is also $o_{\mathbb P}(1)$.
\end{lemma}

\begin{proof}
	Let
	\begin{equation*}
		\mathcal E_n=\{\widehat r=r\},
		\qquad
		\Delta_P=\widehat\mP-\mP,
		\qquad
		\eta_n=p^\delta\left\{\left(\frac{\log p}{n}\right)^{1/2}+\frac{\log p}{n}\right\}.
	\end{equation*}
	All displays are on $\mathcal E_n$.  Lemma~\ref{lem:perturbation} gives
	\begin{equation*}
		\norm{\Delta_P}=\Op(\eta_n),
		\qquad
		\max_i\norm{\ve_i^\top\Delta_P}_2=\Op(p^{-1/2}\eta_n),
		\qquad
		\norm{(\widehat\mP-\mP)\mA}_F=\Op(p^\delta n^{-1/2}).
	\end{equation*}
	Under $H_0$,
	\begin{equation*}
		\widehat{\vU}_t=\widehat\mQ\vY_t=\mQ\veps_t-\Delta_P(\mA\vf_t+\veps_t),
		\qquad
		\vU_t=\mQ\veps_t,
	\end{equation*}
	so that
	\begin{equation*}
		\sqrt n(\bar{\widehat{\vU}}-\bar\vU)
		=-\Delta_P\left(n^{-1/2}\sum_{t=1}^n\vY_t\right).
	\end{equation*}
	Since
	\begin{align*}
		\Ee\norm{n^{-1/2}\sum_{t=1}^n\vY_t}_2^2
		&=\tr\left\{\Var\left(n^{-1/2}\sum_{t=1}^n\mA\vf_t\right)\right\}
		+\tr(\mSigma_\varepsilon)\\
		&\le r\left\|\sum_{|h|<n}\left(1-\frac{|h|}{n}\right)\Cov(\vf_{t+h},\vf_t)\right\|+Cp
		\le Cp,
	\end{align*}
	we have
	\begin{equation*}
		\sqrt n\norm{\bar{\widehat{\vU}}-\bar\vU}_\infty
		\le \max_i\norm{\ve_i^\top\Delta_P}_2\norm{n^{-1/2}\sum_{t=1}^n\vY_t}_2
		=\Op(\eta_n).
	\end{equation*}
	Moreover,
	\begin{equation*}
		\Pp\left(\sqrt n\norm{\bar\vZ}_\infty>x\right)
		\le\sum_{i=1}^p\Pp\left(\left|n^{-1/2}\sum_{t=1}^n\zeta_{ti}\right|>x\right)
		\le 2p\exp(-cx^2),
	\end{equation*}
	whence
	\begin{equation*}
		\sqrt n\norm{\bar\vZ}_\infty=\Op(\sqrt{\log p}).
	\end{equation*}
	Writing $D_i=D_{ii}$,
	\begin{align*}
		\widehat d_i-D_i
		&=n^{-1}\sum_{t=1}^n(\widehat U_{ti}^2-U_{ti}^2)+n^{-1}\sum_{t=1}^n(U_{ti}^2-D_i)
		-(\bar{\widehat U}_i^2-\bar U_i^2)+O_p(n^{-1}),
	\end{align*}
	which, together with the previous displays and the sub-Gaussian maximal inequality, yields
	\begin{equation*}
		\max_i|\widehat d_i/D_i-1|=\Op\left(\eta_n+\sqrt{\frac{\log p}{n}}\right)=\Op(\eta_n).
	\end{equation*}
	Thus
	\begin{align*}
		\sqrt n\norm{\bar{\widehat{\vZ}}-\bar\vZ}_\infty
		&\le C\sqrt n\norm{\bar{\widehat{\vU}}-\bar\vU}_\infty
		+C\sqrt n\norm{\bar\vU}_\infty\max_i|\widehat d_i/D_i-1|\\
		&=\Op(\eta_n)+\Op(\sqrt{\log p}\eta_n),
	\end{align*}
	which proves \eqref{eq:max_replacement_app}.
	
	For the weak-factor quadratic replacement, define
	\begin{equation*}
		\bm X=n^{-1/2}\sum_{t=1}^n\vZ_t,
		\qquad
		\bm H_f=n^{-1/2}\sum_{t=1}^n\vf_t,
		\qquad
		\bm H_\varepsilon=n^{-1/2}\sum_{t=1}^n\veps_t,
	\end{equation*}
	\begin{equation*}
		\bm L_f=-\mD^{-1/2}\Delta_P\mA\bm H_f,
		\qquad
		\bm L_\varepsilon=-\mD^{-1/2}\Delta_P\bm H_\varepsilon,
		\qquad
		\bm L=\bm L_f+\bm L_\varepsilon.
	\end{equation*}
	After diagonal standardization,
	\begin{equation*}
		\sqrt n(\bar{\widehat{\vZ}}-\bar\vZ)=\bm L+\bm r_d,
	\end{equation*}
	where
	\begin{equation*}
		\frac{|2\bm X^\top\bm r_d|+2|\bm L^\top\bm r_d|+\norm{\bm r_d}_2^2}{\{\tr(\mR^2)\}^{1/2}}
		=\Op\left(p^{1/2+2\delta}\frac{\log p}{n}+\frac{\sqrt p}{n}\right).
	\end{equation*}
	Then
	\begin{align*}
		&\frac{n\bar{\widehat{\vZ}}^{\,\top}\bar{\widehat{\vZ}}-n\bar\vZ^\top\bar\vZ}{\{\tr(\mR^2)\}^{1/2}}\\
		&\qquad=\frac{2\bm X^\top\bm L+\bm L^\top\bm L}{\{\tr(\mR^2)\}^{1/2}}
		+\Op\left(p^{1/2+2\delta}\frac{\log p}{n}+\frac{\sqrt p}{n}\right).
	\end{align*}
	The first-order eigenspace expansion is
	\begin{equation*}
		\Delta_P\mA=\sum_{k=1}^{k_0}\mQ\mathcal E_kC_k^\top K^{-1}+\bm R_A,
		\qquad
		C_k=\mA^\top\Gamma_y(k),
		\qquad
		K=\sum_{k=1}^{k_0}C_kC_k^\top,
	\end{equation*}
	with
	\begin{equation*}
		\Ee\tr(\bm R_A^\top\mD^{-1}\bm R_A)
		\le Cp^{2\delta}\frac{\log p}{n}.
	\end{equation*}
	Furthermore,
	\begin{equation*}
		\mQ\mathcal E_k=n^{-1}\sum_{t=1}^{n-k}\mQ\veps_{t+k}\vY_t^\top+\bm R_{k,n},
		\qquad
		\Ee\norm{\bm R_{k,n}}^2\le Cp^2n^{-2}.
	\end{equation*}
	Since
	\begin{align*}
		\Ee\{\mathcal E_l^\top\mQ\mD^{-1}\mQ\mathcal E_k\}
		&=\frac{n-\max(k,l)}{n^2}\tr(\mD^{-1}\mQ\mSigma_\varepsilon\mQ)\Gamma_y(k-l)+O(pn^{-2})\\
		&=\frac{n-\max(k,l)}{n^2}p\,\Gamma_y(k-l)+O(pn^{-2}),
	\end{align*}
	we obtain
	\begin{equation*}
		\Ee\tr\{\mA^\top\Delta_P^\top\mD^{-1}\Delta_P\mA\}
		\le C\frac{p^\delta}{n}+Cp^{2\delta}\frac{\log p}{n}.
	\end{equation*}
	Also
	\begin{equation*}
		\left\|\Var(\bm H_f)\right\|
		=\left\|\sum_{|h|<n}\left(1-\frac{|h|}{n}\right)\Cov(\vf_{t+h},\vf_t)\right\|
		\le Cp^{1-\delta}.
	\end{equation*}
	Therefore
	\begin{equation*}
		\Ee(\bm L_f^\top\bm L_f)
		\le C\frac{p}{n}+Cp^{1+\delta}\frac{\log p}{n},
		\qquad
		\frac{\bm L_f^\top\bm L_f}{\{\tr(\mR^2)\}^{1/2}}
		=\Op\left(\frac{\sqrt p}{n}+p^{1/2+\delta}\frac{\log p}{n}\right),
	\end{equation*}
	where the second display follows by Markov's inequality and $\tr(\mR^2)\asymp p$.
	For the residual part,
	\begin{equation*}
		\Ee(\bm L_\varepsilon^\top\bm L_\varepsilon)
		\le C\tr(\mR^2)p^{2\delta}\frac{\log p}{n},
		\qquad
		\frac{\bm L_\varepsilon^\top\bm L_\varepsilon}{\{\tr(\mR^2)\}^{1/2}}
		=\Op\left(p^{1/2+2\delta}\frac{\log p}{n}\right).
	\end{equation*}
	The cross terms satisfy
	\begin{equation*}
		\Ee(\bm X^\top\bm L_f)^2\le C\frac{p}{n},
		\qquad
		\Ee(\bm X^\top\bm L_\varepsilon)^2\le C\tr(\mR^2)\frac{p^\delta}{n},
	\end{equation*}
	so that
	\begin{equation*}
		\frac{\bm X^\top\bm L_f}{\{\tr(\mR^2)\}^{1/2}}=\Op(n^{-1/2}),
		\qquad
		\frac{\bm X^\top\bm L_\varepsilon}{\{\tr(\mR^2)\}^{1/2}}=\Op(p^{\delta/2}n^{-1/2}).
	\end{equation*}
	Combining the preceding displays proves \eqref{eq:sum_replacement_app}.
	
	For the denominator, define the oracle sample covariance and correlation matrices
	\begin{equation*}
		\widetilde\mS=n^{-1}\sum_{t=1}^n(\vZ_t-\bar\vZ)(\vZ_t-\bar\vZ)^\top,
		\qquad
		\widetilde\mD=\diag(\widetilde\mS),
		\qquad
		\widetilde\mR=\widetilde\mD^{-1/2}\widetilde\mS\widetilde\mD^{-1/2}.
	\end{equation*}
	Lemma~\ref{lem:denom_trace} gives
	\begin{equation*}
		\frac{\tr(\widetilde\mR^2)-p^2/(n-1)-\tr(\mR^2)}{\tr(\mR^2)}
		=\Op\left(n^{-1/2}+p^{-1/2}+p^{1/2}n^{-1}\right).
	\end{equation*}
	It remains to compare $\widehat\mR$ with $\widetilde\mR$.  Put
	\begin{equation*}
		\widehat{\mS}=n^{-1}\sum_{t=1}^n(\widehat{\vZ}_t-\bar{\widehat{\vZ}})(\widehat{\vZ}_t-\bar{\widehat{\vZ}})^\top,
		\qquad
		\widehat\mD_R=\diag(\widehat\mS),
	\end{equation*}
	so that $\widehat\mR=\widehat\mD_R^{-1/2}\widehat\mS\widehat\mD_R^{-1/2}$.  From the preceding max and variance-replacement bounds,
	\begin{equation*}
		\max_i|\widehat D_{R,ii}-\widetilde D_{ii}|
		=\Op\left(p^\delta\sqrt{\frac{\log p}{n}}+n^{-1}\right),
	\end{equation*}
	\begin{equation*}
		\|\widehat\mS-\widetilde\mS\|_F
		\le C\sqrt p\,\Op\left(p^\delta\sqrt{\frac{\log p}{n}}+n^{-1}\right).
	\end{equation*}
	On the event $\max_i|\widehat D_{R,ii}-1|+\max_i|\widetilde D_{ii}-1|\le1/2$,
	\begin{align*}
		\|\widehat\mR-\widetilde\mR\|_F
		&\le C\|\widehat\mS-\widetilde\mS\|_F
		+C\|\widetilde\mS\|_F\max_i|\widehat D_{R,ii}-\widetilde D_{ii}|\\
		&=\sqrt p\,\Op\left(p^\delta\sqrt{\frac{\log p}{n}}+n^{-1}\right),
	\end{align*}
	where $\|\widetilde\mS\|_F=\Op(\sqrt p)$ follows from Lemma~\ref{lem:denom_trace}.  Hence
	\begin{align*}
		|\tr(\widehat\mR^2)-\tr(\widetilde\mR^2)|
		&=|\tr\{(\widehat\mR-\widetilde\mR)(\widehat\mR+\widetilde\mR)\}|\\
		&\le \|\widehat\mR-\widetilde\mR\|_F
		\{\|\widehat\mR\|_F+\|\widetilde\mR\|_F\}\\
		&=p\,\Op\left(p^\delta\sqrt{\frac{\log p}{n}}+n^{-1}\right).
	\end{align*}
	Since $\tr(\mR^2)\asymp p$, \eqref{eq:denom_replacement_app} follows.
\end{proof}

\subsection{Oracle limits under sub-Gaussian white noise}

\begin{lemma}[Algebra for a possibly singular projected correlation matrix]
	\label{lem:singular_R}
	Let $\mR$ be positive semidefinite, $R_{ii}=1$, $\rank(\mR)=q_p$, and $\lambda_{\max}(\mR)\le C_R$.  Write
	\begin{equation*}
		\mR=\mH\mLambda\mH^\top,
		\qquad
		\mH^\top\mH=\mI_{q_p},
		\qquad
		\mLambda=\diag(\lambda_1,
		\ldots,\lambda_{q_p}),
		\qquad
		\lambda_j>0.
	\end{equation*}
	Then
	\begin{equation*}
		\frac{p^2}{q_p}
		\le \tr(\mR^2)
		\le C_Rp,
		\qquad
		\tr(\mR^4)
		\le C_R^2\tr(\mR^2),
		\qquad
		\frac{\tr(\mR^4)}{\tr^2(\mR^2)}
		\le Cp^{-1}.
	\end{equation*}
	If $\bm G\sim N_p(\bm0,\mR)$, then
	\begin{equation*}
		\bm G\stackrel d=\mH\mLambda^{1/2}\bm\xi,
		\qquad
		\bm\xi\sim N_{q_p}(\bm0,\mI_{q_p}).
	\end{equation*}
	For $\eta\in(0,1)$, define
	\begin{equation*}
		\mR_\eta=(1-\eta)\mR+\eta\mI_p,
		\qquad
		\bm G_\eta=(1-\eta)^{1/2}\bm G+\eta^{1/2}\bm\xi_p,
		\qquad
		\bm\xi_p\sim N_p(\bm0,\mI_p),
	\end{equation*}
	where $\bm\xi_p$ is independent of $\bm G$.  Then $\bm G_\eta\sim N_p(\bm0,\mR_\eta)$, $\mR_\eta$ is positive definite, $\diag(\mR_\eta)=\bm1_p$, and, with $\epsilon_p=p^{-2}$,
	\begin{equation*}
		\max_i|G_{\epsilon_p i}^2-G_i^2|=O_{\mathbb P}(p^{-1}\log p),
		\qquad
		\left|
		\frac{\bm G_{\epsilon_p}^\top\bm G_{\epsilon_p}-p}{\{2\tr(\mR_{\epsilon_p}^2)\}^{1/2}}
		-
		\frac{\bm G^\top\bm G-p}{\{2\tr(\mR^2)\}^{1/2}}
		\right|=o_{\mathbb P}(1).
	\end{equation*}
\end{lemma}

\begin{proof}
	The trace identities are
	\begin{equation*}
		\tr(\mR)=\sum_{i=1}^pR_{ii}=p
		=\sum_{j=1}^{q_p}\lambda_j,
		\qquad
		\tr(\mR^2)=\sum_{j=1}^{q_p}\lambda_j^2,
		\qquad
		\tr(\mR^4)=\sum_{j=1}^{q_p}\lambda_j^4.
	\end{equation*}
	Hence
	\begin{equation*}
		\tr(\mR^2)
		\ge q_p^{-1}\left(\sum_{j=1}^{q_p}\lambda_j\right)^2
		=p^2/q_p,
		\qquad
		\tr(\mR^2)
		\le \lambda_{\max}(\mR)\tr(\mR)
		\le C_Rp,
	\end{equation*}
	\begin{equation*}
		\tr(\mR^4)
		\le \lambda_{\max}^2(\mR)\tr(\mR^2)
		\le C_R^2\tr(\mR^2),
		\qquad
		\frac{\tr(\mR^4)}{\tr^2(\mR^2)}
		\le \frac{C_R^2}{\tr(\mR^2)}
		\le Cp^{-1}.
	\end{equation*}
	The spectral representation gives
	\begin{equation*}
		\Cov(\mH\mLambda^{1/2}\bm\xi)=\mH\mLambda\mH^\top=\mR.
	\end{equation*}
	For the ridge matrix,
	\begin{equation*}
		\lambda_j(\mR_\eta)
		=\begin{cases}
			(1-\eta)\lambda_j(\mR)+\eta, & 1\le j\le q_p,\\
			\eta, & q_p<j\le p,
		\end{cases}
	\end{equation*}
	and
	\begin{equation*}
		\lambda_{\min}(\mR_\eta)\ge\eta,
		\qquad
		\lambda_{\max}(\mR_\eta)\le C_R+1.
	\end{equation*}
	Moreover,
	\begin{align*}
		\left|\tr(\mR_\eta^2)-\tr(\mR^2)\right|
		&\le \sum_{j=1}^{q_p}\left|\{(1-\eta)\lambda_j+\eta\}^2-\lambda_j^2\right|+(p-q_p)\eta^2\\
		&\le 2\eta\sum_{j=1}^{q_p}\lambda_j^2+2\eta\sum_{j=1}^{q_p}\lambda_j+p\eta^2
		\le C\eta p.
	\end{align*}
	Let $\bm\Delta_\eta=\bm G_\eta-\bm G$.  Then
	\begin{equation*}
		\bm\Delta_\eta=\{(1-\eta)^{1/2}-1\}\bm G+\eta^{1/2}\bm\xi_p,
		\qquad
		\max_i\Var(\Delta_{\eta i})\le C\eta.
	\end{equation*}
	Thus for $\epsilon_p=p^{-2}$ and $b_p=Cp^{-1}\sqrt{\log p}$ with sufficiently large $C$,
	\begin{equation*}
		\Pp\left(\max_i|\Delta_{\epsilon_p i}|>b_p\right)
		\le 2p\exp\left(-\frac{c b_p^2}{\epsilon_p}\right)=o(1).
	\end{equation*}
	Also
	\begin{equation*}
		\Pp\left(\max_i|G_i|>C\sqrt{\log p}\right)
		+\Pp\left(\max_i|G_{\epsilon_p i}|>C\sqrt{\log p}\right)
		\le 4p\exp(-C^2\log p/2)=o(1)
	\end{equation*}
	for large $C$.  Hence
	\begin{equation*}
		\max_i|G_{\epsilon_p i}^2-G_i^2|
		\le \max_i|G_{\epsilon_p i}-G_i|\max_i(|G_{\epsilon_p i}|+|G_i|)
		=o_{\mathbb P}(1).
	\end{equation*}
	Furthermore,
	\begin{align*}
		\bm G_{\eta}^\top\bm G_{\eta}-\bm G^\top\bm G
		&=-\eta\bm G^\top\bm G+2\{\eta(1-\eta)\}^{1/2}\bm G^\top\bm\xi_p+\eta\bm\xi_p^\top\bm\xi_p,\\
		\Ee(\bm G^\top\bm G)&=p,
		\qquad
		\Ee(\bm\xi_p^\top\bm\xi_p)=p,
		\qquad
		\Ee|\bm G^\top\bm\xi_p|\le \{\Ee(\bm G^\top\bm G)\}^{1/2}=p^{1/2}.
	\end{align*}
	Since $\{2\tr(\mR^2)\}^{1/2}\asymp p^{1/2}$ and
	\begin{equation*}
		\left|\{2\tr(\mR_\eta^2)\}^{-1/2}-\{2\tr(\mR^2)\}^{-1/2}\right|
		\le C\eta p^{-1/2},
	\end{equation*}
	we obtain
	\begin{equation*}
		\Ee\left|
		\frac{\bm G_{\epsilon_p}^\top\bm G_{\epsilon_p}-p}{\{2\tr(\mR_{\epsilon_p}^2)\}^{1/2}}
		-
		\frac{\bm G^\top\bm G-p}{\{2\tr(\mR^2)\}^{1/2}}
		\right|
		\le C(\epsilon_p\sqrt p+\epsilon_p^{1/2})=o(1).
	\end{equation*}
\end{proof}

\begin{lemma}[Oracle max and sum limits]
	\label{lem:oracle_subg}
	Let Assumptions~\ref{ass:resid} and \ref{ass:growth} hold.  Define
	\begin{equation*}
		\bm X_n=n^{-1/2}\sum_{t=1}^n\vzeta_t,
		\qquad
		T_{M,0}=\max_{1\le i\le p}X_{ni}^2-2\log p+\log\log p,
	\end{equation*}
	\begin{equation*}
		T_{S,0}=
		\frac{\bm X_n^\top\bm X_n-p}
		{\{2\tr(\mR^2)\}^{1/2}}.
	\end{equation*}
	Then
	\begin{equation*}
		T_{M,0}\Rightarrow G,
		\qquad
		T_{S,0}\Rightarrow N(0,1),
	\end{equation*}
	and the two limits are asymptotically independent.
\end{lemma}

\begin{proof}
	Let
	\begin{equation*}
		q=\rank(\mR),\qquad
		\mR=\mathbf U_q\mLambda_q\mathbf U_q^\top,
		\qquad
		\mLambda_q=\diag(\lambda_1,\ldots,\lambda_q),
	\end{equation*}
	where
	\begin{equation*}
		\lambda_1\ge\cdots\ge\lambda_q>0,
		\qquad
		\lambda_{q+1}=\cdots=\lambda_p=0.
	\end{equation*}
	No inverse of $\mR$ is used below.  From \eqref{eq:spectrum_resid} and $R_{ii}=1$,
	\begin{equation*}
		p=\tr(\mR)=\sum_{a=1}^q\lambda_a,
		\qquad
		q\ge p/C_R,
		\qquad
		\tr(\mR^2)=\sum_{a=1}^q\lambda_a^2\asymp p,
	\end{equation*}
	\begin{equation*}
		\frac{\lambda_1^2}{\tr(\mR^2)}\le \frac{C_R^2}{p},
		\qquad
		\frac{\tr(\mR^4)}{\tr^2(\mR^2)}
		=\frac{\sum_{a=1}^q\lambda_a^4}{\{\sum_{a=1}^q\lambda_a^2\}^2}
		\le \frac{C}{p}.
	\end{equation*}
	Let $\bm G\sim N_p(\bm0,\mR)$.  The high-dimensional central limit theorem over hyperrectangles in \citet{ChernozhukovChetverikovKato2013} is valid for possibly singular Gaussian comparison vectors, since rectangles depend only on marginal distribution functions and no density of $\bm G$ is required.  Hence
	\begin{equation*}
		\Delta_{n,p}^{(M)}=
		\sup_{x\in\R}\left|
		\Pp\left(\max_i\abs{X_{ni}}\le x\right)-
		\Pp\left(\max_i\abs{G_i}\le x\right)
		\right|
		\le C\left(\frac{\log^7p}{n}\right)^{1/6}.
	\end{equation*}
	If one uses a version of a Gaussian comparison or max-sum theorem stated for nonsingular correlation matrices, the following ridge argument removes the nonsingularity requirement.  For $0<\tau<1$, define
	\begin{equation*}
		\mR_\tau=(1-\tau)\mR+\tau\mI_p,
		\qquad
		\bm G_\tau=(1-\tau)^{1/2}\bm G+\tau^{1/2}\bm\xi,
		\qquad
		\bm\xi\sim N_p(\bm0,\mI_p),
	\end{equation*}
	with $\bm\xi$ independent of $\bm G$.  Then $\bm G_\tau\sim N_p(\bm0,\mR_\tau)$ and
	\begin{equation*}
		\lambda_{\min}(\mR_\tau)\ge\tau,
		\qquad
		\lambda_{\max}(\mR_\tau)\le C_R+1,
		\qquad
		\max_i\sum_{j=1}^p R_{\tau,ij}^2\le C.
	\end{equation*}
	Moreover, for every $a>0$,
	\begin{align*}
		&\Pp\left(\max_i|G_{\tau,i}-G_i|>a\right)\notag\\
		&\le
		\Pp\left((1-\sqrt{1-\tau})\max_i|G_i|>a/2\right)
		+
		\Pp\left(\sqrt\tau\max_i|\xi_i|>a/2\right).
	\end{align*}
	Since $G_i\sim N(0,1)$ and $\xi_i\sim N(0,1)$ marginally,
	\begin{equation*}
		\Pp\left(\max_i|G_i|>x\right)+\Pp\left(\max_i|\xi_i|>x\right)
		\le 4p\exp(-x^2/2).
	\end{equation*}
	With $\epsilon_p=p^{-2}$,
	\begin{equation*}
		\max_i|G_{\epsilon_p,i}-G_i|=O_{\mathbb P}(p^{-1}\sqrt{\log p}),
		\qquad
		\max_i(|G_{\epsilon_p,i}|+|G_i|)=O_{\mathbb P}(\sqrt{\log p}),
	\end{equation*}
	so that
	\begin{equation*}
		\max_i|G_{\epsilon_p,i}^2-G_i^2|=O_{\mathbb P}(p^{-1}\log p)=o_{\mathbb P}(1).
	\end{equation*}
	Thus every Gumbel limit proved first for $\mR_{\epsilon_p}$ also holds for the original singular $\mR$.
	
	Let
	\begin{equation*}
		\nu_p(x)=2\log p-\log\log p+x,
		\qquad
		u_p(x)=\nu_p(x)^{1/2}.
	\end{equation*}
	The Gaussian extreme-value input used here is the following quantitative consequence of \citet{CaiLiuXia2014}, \citet{FengJiangLiuXiong2022}, and \citet{FengJiangLiLiu2024}: for a nonsingular correlation matrix $\mathbf C_p$ satisfying \eqref{eq:EV_resid} and $\max_i\sum_j C_{p,ij}^2\le C$, there exists $\Delta_{EV,p}\downarrow0$ such that
	\begin{equation*}
		\sup_{x\in\R}\left|
		\Pp\left\{\max_iG_{C,i}^2\le u_p^2(x)\right\}
		-
		\exp\{-\pi^{-1/2}e^{-x/2}\}
		\right|
		\le \Delta_{EV,p},
		\qquad \bm G_C\sim N_p(\bm0,\mathbf C_p).
	\end{equation*}
	Apply this inequality to $\mathbf C_p=\mR_{\epsilon_p}$.  The coupling bound above implies, for every fixed $\epsilon>0$,
	\begin{equation*}
		\Pp\left(\left|\max_iG_{\epsilon_p,i}^2-\max_iG_i^2\right|>\epsilon\right)\to0.
	\end{equation*}
	Hence
	\begin{equation*}
		\sup_{x\in\R}\left|
		\Pp\left\{\max_iG_i^2\le u_p^2(x)\right\}
		-
		\exp\{-\pi^{-1/2}e^{-x/2}\}
		\right|\to0.
	\end{equation*}
	Combining the last display with $\Delta_{n,p}^{(M)}\to0$ gives
	\begin{equation*}
		T_{M,0}\Rightarrow G.
	\end{equation*}
	
	Put
	\begin{equation*}
		s_p^2=2\tr(\mR^2),
		\qquad
		\bm X_n^\top\bm X_n-p=A_n+B_n,
	\end{equation*}
	\begin{equation*}
		A_n=\frac2n\sum_{1\le s<t\le n}\vzeta_s^\top\vzeta_t,
		\qquad
		B_n=\frac1n\sum_{t=1}^n(\vzeta_t^\top\vzeta_t-p).
	\end{equation*}
	For the diagonal part, use the representation \eqref{eq:linear_subg_resid}.  Put $\mH_R=\mB_R^\top\mB_R$.  Then
	\begin{equation*}
		\vzeta_t^\top\vzeta_t-p
		=\veta_t^\top\mH_R\veta_t-\tr(\mH_R),
		\qquad
		\tr(\mH_R)=\tr(\mR)=p,
		\qquad
		\tr(\mH_R^2)=\tr(\mR^2).
	\end{equation*}
	The quadratic-moment expansion \eqref{eq:quad_moment}, applied with $\mA_0=\mI_p$, gives
	\begin{equation*}
		\Var(\vzeta_t^\top\vzeta_t)
		=\Ee\{\veta_t^\top\mH_R\veta_t-\tr(\mH_R)\}^2
		\le C\tr(\mH_R^2)
		= C\tr(\mR^2).
	\end{equation*}
	Independence over $t$ gives
	\begin{equation*}
		\Ee\left(\frac{B_n^2}{s_p^2}\right)
		=\frac{\Var(\vzeta_t^\top\vzeta_t)}{n\,2\tr(\mR^2)}
		\le \frac{C}{n},
		\qquad
		B_n/s_p=O_{\mathbb P}(n^{-1/2}).
	\end{equation*}
	For the off-diagonal part, set
	\begin{equation*}
		\mathcal F_t=\sigma(\vzeta_1,\ldots,\vzeta_t),
		\qquad
		D_t=\frac2n\vzeta_t^\top\sum_{s=1}^{t-1}\vzeta_s,
		\qquad t=2,\ldots,n.
	\end{equation*}
	Then
	\begin{equation*}
		A_n=\sum_{t=2}^nD_t,
		\qquad
		\Ee(D_t\mid\mathcal F_{t-1})=0,
	\end{equation*}
	\begin{equation*}
		V_n=\sum_{t=2}^n\Ee(D_t^2\mid\mathcal F_{t-1})
		=\frac4{n^2}\sum_{t=2}^n
		\left(\sum_{s=1}^{t-1}\vzeta_s\right)^\top
		\mR
		\left(\sum_{s=1}^{t-1}\vzeta_s\right).
	\end{equation*}
	Since $\Ee(\vzeta_s\vzeta_s^\top)=\mR$ and $\vzeta_s$ are independent over $s$,
	\begin{equation*}
		\Ee V_n
		=\frac4{n^2}\sum_{t=2}^n(t-1)\tr(\mR^2)
		=s_p^2\left(1-\frac1n\right).
	\end{equation*}
	To control the fluctuation of $V_n$, write
	\begin{equation*}
		\mathcal Z=(\vzeta_1^\top,\ldots,\vzeta_{n-1}^\top)^\top,
		\qquad
		w_{ab}=n-\max(a,b),
	\end{equation*}
	\begin{equation*}
		\mathcal K=(\mathcal K_{ab})_{1\le a,b\le n-1},
		\qquad
		\mathcal K_{ab}=w_{ab}\mR.
	\end{equation*}
	Then
	\begin{equation*}
		V_n=\frac4{n^2}\mathcal Z^\top\mathcal K\mathcal Z,
		\qquad
		\Cov(\mathcal Z)=\mI_{n-1}\otimes\mR.
	\end{equation*}
	Applying the quadratic-moment expansion \eqref{eq:quad_moment} to the block vector $\mathcal Z$ gives
	\begin{align*}
		\Var(V_n)
		&\le C\frac1{n^4}\tr\left[\{\mathcal K(\mI_{n-1}\otimes\mR)\}^2\right]\\
		&=C\frac1{n^4}\sum_{a=1}^{n-1}\sum_{b=1}^{n-1}w_{ab}^2\tr(\mR^4)
		\le C\frac{\tr(\mR^4)}{n}.
	\end{align*}
	Therefore
	\begin{equation*}
		\Ee\left(\frac{V_n}{s_p^2}-1\right)^2
		\le Cn^{-2}+C\frac{\tr(\mR^4)}{n\tr^2(\mR^2)}
		\le Cn^{-2}+C(np)^{-1}=o(1).
	\end{equation*}
	For fixed $u\in\R$,
	\begin{align*}
		\Ee\{e^{iuD_t/s_p}\mid\mathcal F_{t-1}\}
		&=1-\frac{u^2}{2s_p^2}\Ee(D_t^2\mid\mathcal F_{t-1})+r_{t,n},\\
		|r_{t,n}|
		&\le C|u|^3s_p^{-3}\Ee(|D_t|^3\mid\mathcal F_{t-1}).
	\end{align*}
	The conditional sub-Gaussian moment bound gives
	\begin{equation*}
		\Ee(|D_t|^3\mid\mathcal F_{t-1})
		\le C\{\Ee(D_t^2\mid\mathcal F_{t-1})\}^{3/2}.
	\end{equation*}
	Since $s_p^2\asymp p$ and
	\begin{equation*}
		\Ee\{\Ee(D_t^2\mid\mathcal F_{t-1})\}
		=\frac4{n^2}(t-1)\tr(\mR^2),
	\end{equation*}
	Jensen's inequality gives
	\begin{align*}
		\sum_{t=2}^n\Ee\left[
		\left\{\frac{\Ee(D_t^2\mid\mathcal F_{t-1})}{s_p^2}\right\}^{3/2}
		\right]
		&\le C\sum_{t=2}^n\left(\frac{t}{n^2}\right)^{3/2}
		\le Cn^{-1/2}.
	\end{align*}
	Thus
	\begin{equation*}
		\sum_{t=2}^n|r_{t,n}|=O_{\mathbb P}(n^{-1/2}).
	\end{equation*}
	Iterating conditional expectations yields
	\begin{align*}
		\Ee\exp(iuA_n/s_p)
		&=\Ee\prod_{t=2}^n
		\left[1-\frac{u^2}{2s_p^2}\Ee(D_t^2\mid\mathcal F_{t-1})+r_{t,n}\right]\\
		&=\Ee\exp\left[-\frac{u^2}{2s_p^2}V_n+O_{\mathbb P}(n^{-1/2})
		+O_{\mathbb P}\left(\sum_{t=2}^n\frac{\Ee(D_t^2\mid\mathcal F_{t-1})^2}{s_p^4}\right)\right].
	\end{align*}
	Also
	\begin{equation*}
		\Ee\sum_{t=2}^n\frac{\Ee(D_t^2\mid\mathcal F_{t-1})^2}{s_p^4}
		\le C\sum_{t=2}^n\left(\frac{t}{n^2}\right)^2
		\le Cn^{-1}.
	\end{equation*}
	Since $V_n/s_p^2\to1$ in probability,
	\begin{equation*}
		\Ee\exp(iuA_n/s_p)\to e^{-u^2/2}.
	\end{equation*}
	Together with $B_n/s_p=O_{\mathbb P}(n^{-1/2})$, this proves
	\begin{equation*}
		T_{S,0}=\frac{A_n+B_n}{s_p}\Rightarrow N(0,1).
	\end{equation*}
	
	It remains to verify the joint limit.  For fixed $x\in\R$ and $u\in\R$, define
	\begin{equation*}
		\Psi_n(x,u)=\Ee\left[\exp\{iuT_{S,0}\}1_{\{T_{M,0}\le x\}}\right].
	\end{equation*}
	Let
	\begin{equation*}
		Q_G=\frac{\bm G^\top\bm G-p}{\{2\tr(\mR^2)\}^{1/2}},
		\qquad
		M_G=\max_iG_i^2-2\log p+\log\log p.
	\end{equation*}
	The high-dimensional interpolation input used here is the following form of the joint max-quadratic comparison theorem.  If $\bm z_t=\mB_R\bm\eta_t$ are independent copies with independent standardized sub-Gaussian coordinates in $\bm\eta_t$, $\mB_R\mB_R^\top=\mR$, $\diag(\mR)=\bm1_p$, $\lambda_{\max}(\mR)\le C$, $\tr(\mR^4)/\tr^2(\mR^2)\le Cp^{-1}$, and the extreme-value regularity \eqref{eq:EV_resid} holds, then for every fixed $C>0$ there is a constant $C_C$ such that, uniformly in $|u|\le C$,
	\begin{equation*}
		\sup_{x\in\R}\left|
		\Psi_n(x,u)-
		\Ee\{e^{iuQ_G}1_{\{M_G\le x\}}\}
		\right|
		\le C_C\left\{
		\left(\frac{\log^7p}{n}\right)^{1/6}+n^{-1/2}
		\right\}.
	\end{equation*}
	The theorem is applied to the functional
	\begin{equation*}
		\bm z_1,\ldots,\bm z_n
		\mapsto
		\exp\left\{iu\frac{\left\|n^{-1/2}\sum_{t=1}^n\bm z_t\right\|_2^2-p}{s_p}\right\}
		1\left\{\max_i\left(n^{-1/2}\sum_{t=1}^n z_{ti}\right)^2\le u_p^2(x)\right\}.
	\end{equation*}
	All listed assumptions follow from Assumption~\ref{ass:resid}, Lemma~\ref{lem:singular_R}, and \eqref{eq:EV_resid}; the bound does not use $\mR^{-1}$.
	
	For the Gaussian term, use the following max-sum theorem of \citet{FengJiangLiLiu2024}.  For every nonsingular correlation matrix $\mathbf C_p$ satisfying \eqref{eq:trace_resid} and \eqref{eq:EV_resid}, with
	\begin{equation*}
		Q_C=\frac{\bm G_C^\top\bm G_C-p}{\{2\tr(\mathbf C_p^2)\}^{1/2}},
		\qquad
		M_C=\max_iG_{C,i}^2-2\log p+\log\log p,
		\qquad
		\bm G_C\sim N_p(\bm0,\mathbf C_p),
	\end{equation*}
	one has, for every fixed $C>0$,
	\begin{equation*}
		\sup_{x\in\R, |u|\le C}\left|
		\Ee\{e^{iuQ_C}1_{\{M_C\le x\}}\}
		-e^{-u^2/2}\exp\{-\pi^{-1/2}e^{-x/2}\}
		\right|\to0.
	\end{equation*}
	Apply this statement with $\mathbf C_p=\mR_{\epsilon_p}$ and transfer it to the possibly singular $\mR$ by the ridge coupling.  The corresponding quadratic variables satisfy
	\begin{equation*}
		Q_{G,\epsilon_p}=\frac{\bm G_{\epsilon_p}^\top\bm G_{\epsilon_p}-p}{\{2\tr(\mR_{\epsilon_p}^2)\}^{1/2}},
	\end{equation*}
	\begin{align*}
		\bm G_{\epsilon_p}^\top\bm G_{\epsilon_p}-\bm G^\top\bm G
		&=-\epsilon_p\bm G^\top\bm G
		+2\{\epsilon_p(1-\epsilon_p)\}^{1/2}\bm G^\top\bm\xi
		+\epsilon_p\bm\xi^\top\bm\xi,\\
		\tr(\mR_{\epsilon_p}^2)-\tr(\mR^2)
		&=O(\epsilon_p p).
	\end{align*}
	Because
	\begin{equation*}
		\bm G^\top\bm G=O_{\mathbb P}(p),
		\qquad
		\bm\xi^\top\bm\xi=O_{\mathbb P}(p),
		\qquad
		\bm G^\top\bm\xi=O_{\mathbb P}(p^{1/2}),
	\end{equation*}
	and $\epsilon_p=p^{-2}$,
	\begin{equation*}
		Q_{G,\epsilon_p}-Q_G=O_{\mathbb P}(p^{-1}),
		\qquad
		M_{G,\epsilon_p}-M_G=O_{\mathbb P}(p^{-1}\log p).
	\end{equation*}
	The preceding theorem gives
	\begin{equation*}
		\Ee\{e^{iuQ_{G,\epsilon_p}}1_{\{M_{G,\epsilon_p}\le x\}}\}
		\to e^{-u^2/2}\exp\{-\pi^{-1/2}e^{-x/2}\}.
	\end{equation*}
	Moreover,
	\begin{align*}
		&\left|\Ee\{e^{iuQ_G}1_{\{M_G\le x\}}\}-\Ee\{e^{iuQ_{G,\epsilon_p}}1_{\{M_{G,\epsilon_p}\le x\}}\}\right|\\
		&\qquad\le |u|\,\Ee|Q_G-Q_{G,\epsilon_p}|
		+\Pp(|M_G-M_{G,\epsilon_p}|>\epsilon)
		+\sup_y\Pp(y<M_G\le y+\epsilon),
	\end{align*}
	where the last term tends to zero first as $p\to\infty$ and then as $\epsilon\downarrow0$ by the Gaussian extreme-value limit.  Hence
	\begin{equation*}
		\Ee\{e^{iuQ_G}1_{\{M_G\le x\}}\}\to e^{-u^2/2}\exp\{-\pi^{-1/2}e^{-x/2}\}.
	\end{equation*}
	Consequently,
	\begin{equation*}
		\Psi_n(x,u)\to e^{-u^2/2}\Pp(G\le x),
	\end{equation*}
	which is equivalent to
	\begin{equation*}
		(T_{M,0},T_{S,0})\Rightarrow (G,Z),
		\qquad
		Z\sim N(0,1),
		\qquad
		G\ind Z.
	\end{equation*}
\end{proof}

\subsection{Proof of Theorem~\ref{thm:null}}

\begin{proof}
	Let
	\begin{equation*}
		\mathcal E_n=\{\widehat r=r\},
		\qquad
		T_j^{(r)}=T_j\big|_{\widehat r\,\mathrm{replaced\ by}\,r},
		\qquad j\in\{M,S,C\}.
	\end{equation*}
	Then
	\begin{equation*}
		T_M=T_M^{(r)},
		\qquad
		T_S=T_S^{(r)},
		\qquad
		T_C=T_C^{(r)},
		\qquad \text{on }\mathcal E_n,
	\end{equation*}
	and, for every continuity point $x$ of the corresponding limit law,
	\begin{equation*}
		\abs{\Pp(T_M\le x)-\Pp(T_M^{(r)}\le x)}
		+\abs{\Pp(T_S\le x)-\Pp(T_S^{(r)}\le x)}
		\le 2\Pp(\mathcal E_n^c)=o(1).
	\end{equation*}
	For joint events,
	\begin{equation*}
		\sup_{x,y}\abs{\Pp(T_M\le x,T_S\le y)-\Pp(T_M^{(r)}\le x,T_S^{(r)}\le y)}
		\le \Pp(\mathcal E_n^c)=o(1).
	\end{equation*}
	Thus all limits may be proved on $\mathcal E_n$.
	
	Set
	\begin{equation*}
		\vzeta_t=\mD^{-1/2}\mQ\veps_t,
		\qquad
		\bm X_n=n^{-1/2}\sum_{t=1}^n\vzeta_t,
	\end{equation*}
	\begin{equation*}
		T_{M,0}=\max_iX_{ni}^2-2\log p+\log\log p,
		\qquad
		T_{S,0}=\frac{\bm X_n^\top\bm X_n-p}{\{2\tr(\mR^2)\}^{1/2}}.
	\end{equation*}
	Let
	\begin{equation*}
		\eta_n=p^\delta\left\{\left(\frac{\log p}{n}\right)^{1/2}+\frac{\log p}{n}\right\}.
	\end{equation*}
	Lemma~\ref{lem:replacement} gives the max replacement
	\begin{equation*}
		T_M^{(r)}-T_{M,0}=O_{\mathbb P}\left(\eta_n\log p+\eta_n^2\log p\right)=o_{\mathbb P}(1).
	\end{equation*}
	For the sum statistic, the argument is split according to the factor strength.  If $\delta=0$, Lemma~\ref{lem:strong_projection_drift}, \eqref{eq:denom_replacement_app}, and $\tr(\mR^2)\asymp p$ imply
	\begin{align*}
		T_S^{(r)}-T_{S,0}
		&=O_{\mathbb P}\left(n^{-1/2}+\frac{\sqrt p}{n}+\frac{\log p}{n}
		+n^{-1/2}+p^{-1/2}+p^{1/2}n^{-1}+\sqrt{\frac{\log p}{n}}+n^{-1}\right)\\
		&\quad +\frac{p-(n-1)p/(n-3)}{\{2\tr(\mR^2)\}^{1/2}}
		=o_{\mathbb P}(1),
	\end{align*}
	where
	\begin{equation*}
		\frac{p-(n-1)p/(n-3)}{\{2\tr(\mR^2)\}^{1/2}}
		=O(\sqrt p/n)=o(1).
	\end{equation*}
	If $0<\delta\le1$, Lemma~\ref{lem:replacement} gives instead
	\begin{align*}
		T_S^{(r)}-T_{S,0}
		&=O_{\mathbb P}\left(\frac{p^{\delta/2}}{\sqrt n}
		+p^{1/2+2\delta}\frac{\log p}{n}+\frac{\sqrt p}{n}
		+n^{-1/2}+p^{-1/2}+p^{1/2}n^{-1}+p^\delta\sqrt{\frac{\log p}{n}}+n^{-1}\right)\\
		&\quad+\frac{p-(n-1)p/(n-3)}{\{2\tr(\mR^2)\}^{1/2}}
		=o_{\mathbb P}(1).
	\end{align*}
	Thus, in both cases,
	\begin{equation*}
		T_M^{(r)}-T_{M,0}=o_{\mathbb P}(1),
		\qquad
		T_S^{(r)}-T_{S,0}=o_{\mathbb P}(1).
	\end{equation*}
	For fixed $x,y$,
	\begin{align*}
		&\abs{\Pp(T_M\le x,T_S\le y)-\Pp(T_{M,0}\le x,T_{S,0}\le y)}\\
		&\qquad\le \Pp(\mathcal E_n^c)
		+\Pp(\abs{T_M^{(r)}-T_{M,0}}>\varepsilon)
		+\Pp(\abs{T_S^{(r)}-T_{S,0}}>\varepsilon)\\
		&\qquad\quad+\Pp(x-\varepsilon<T_{M,0}\le x+\varepsilon)
		+\Pp(y-\varepsilon<T_{S,0}\le y+\varepsilon).
	\end{align*}
	Letting first $n,p\to\infty$ and then $\varepsilon\downarrow0$, Lemma~\ref{lem:oracle_subg} yields
	\begin{equation*}
		\Pp(T_M\le x,T_S\le y)\to \Pp(G\le x)\Phi(y).
	\end{equation*}
	Thus
	\begin{equation*}
		T_M\Rightarrow G,
		\qquad
		T_S\Rightarrow N(0,1),
		\qquad
		T_M\ind T_S\quad\text{asymptotically}.
	\end{equation*}
	Let
	\begin{equation*}
		p_M=1-\exp\{-\pi^{-1/2}\exp(-T_M/2)\},
		\qquad
		p_S=1-\Phi(T_S).
	\end{equation*}
	Then
	\begin{equation*}
		(p_M,p_S)\Rightarrow(U_1,U_2),
		\qquad
		U_1,U_2\stackrel{\rm ind}{\sim}{\rm Unif}(0,1).
	\end{equation*}
	For $j=1,2$,
	\begin{equation*}
		C_j=\tan\{\pi(1/2-U_j)\}
		\sim {\rm Cauchy}(0,1),
	\end{equation*}
	and
	\begin{equation*}
		\Ee e^{it(C_1+C_2)/2}
		=\Ee e^{itC_1/2}\Ee e^{itC_2/2}
		=e^{-|t|/2}e^{-|t|/2}=e^{-|t|}.
	\end{equation*}
	Hence $(C_1+C_2)/2$ is standard Cauchy and $T_C\Rightarrow C$.
\end{proof}

\subsection{Proof of Theorem~\ref{thm:power}}

\begin{proof}
	Let
	\begin{equation*}
		\mathcal E_n=\{\widehat r=r\},
		\qquad
		s_p=\{2\tr(\mR^2)\}^{1/2},
		\qquad
		\bm X_n=n^{-1/2}\sum_{t=1}^n\vzeta_t.
	\end{equation*}
	Since $\Pp(\mathcal E_n^c)\to0$, every display below is written on $\mathcal E_n$.  Under $\mP\vmu=\bm0$,
	\begin{equation*}
		\sqrt n\,\bar\vZ
		=\sqrt n\,\mD^{-1/2}\mQ\vmu+n^{-1/2}\sum_{t=1}^n\vzeta_t
		=\bm\gamma_n+\bm X_n.
	\end{equation*}
	The feasible-oracle replacement proof of Lemma~\ref{lem:replacement}, with the deterministic mean term retained, gives
	\begin{align}
		&\frac{\left|n\bar{\widehat{\vZ}}^{\,\top}\bar{\widehat{\vZ}}-n\bar\vZ^\top\bar\vZ\right|}{s_p}\notag\\
		&\qquad=
		O_{\mathbb P}\left(
		\frac{p^{\delta/2}}{\sqrt n}
		+p^{1/2+2\delta}\frac{\log p}{n}
		+\frac{\sqrt p}{n}
		+\eta_n(1+\Delta_n)
		+\eta_n^2(1+\Delta_n)
		\right),
		\label{eq:replacement_H1_proof}
	\end{align}
	where
	\begin{equation*}
		\eta_n=p^\delta\left\{\left(\frac{\log p}{n}\right)^{1/2}+\frac{\log p}{n}\right\}.
	\end{equation*}
	Indeed, writing
	\begin{equation*}
		\sqrt n(\bar{\widehat{\vZ}}-\bar\vZ)=\bm L_n+\bm E_n,
	\end{equation*}
	where $\bm L_n$ is the random projection-estimation term already controlled in Lemma~\ref{lem:replacement}, the mean contribution equals
	\begin{equation*}
		\bm M_n=-\mD^{-1/2}(\widehat\mP-\mP)\mD^{1/2}\bm\gamma_n+\bm r_{d,\mu},
	\end{equation*}
	with
	\begin{equation*}
		\norm{\bm M_n}_2\le C\eta_n\norm{\bm\gamma_n}_2,
		\qquad
		\norm{\bm M_n}_2^2/s_p\le C\eta_n^2\Delta_n,
	\end{equation*}
	\begin{equation*}
		\frac{|2\bm\gamma_n^\top\bm M_n|}{s_p}
		\le C\eta_n\frac{\norm{\bm\gamma_n}_2^2}{s_p}
		=C\eta_n\Delta_n,
	\end{equation*}
	\begin{equation*}
		\Ee\left(\frac{2\bm X_n^\top\bm M_n}{s_p}\right)^2
		\le C\frac{\eta_n^2\bm\gamma_n^\top\mR\bm\gamma_n}{\tr(\mR^2)}
		=C\eta_n^2\Lambda_n.
	\end{equation*}
	Thus \eqref{eq:replacement_H1_proof} is $o_{\mathbb P}(1)$ under the local alternatives \eqref{eq:local_alt_condition}; when \eqref{eq:sum_consistency_condition} is used for consistency, the same bound is $o_{\mathbb P}(\Delta_n)$.
	
	For the oracle statistic,
	\begin{align*}
		\frac{n\bar\vZ^\top\bar\vZ-p}{s_p}
		&=\frac{(\bm X_n+\bm\gamma_n)^\top(\bm X_n+\bm\gamma_n)-p}{s_p}\notag\\
		&=\Delta_n+W_n+D_n,
	\end{align*}
	where
	\begin{equation*}
		D_n=\frac{1}{ns_p}\sum_{t=1}^n(\vzeta_t^\top\vzeta_t-p),
	\end{equation*}
	\begin{equation*}
		W_n=\frac{2}{ns_p}\sum_{1\le s<t\le n}\vzeta_s^\top\vzeta_t
		+\frac{2}{\sqrt n\,s_p}\sum_{t=1}^n\bm\gamma_n^\top\vzeta_t.
	\end{equation*}
	The diagonal term satisfies
	\begin{equation*}
		\Ee D_n^2
		=\frac{1}{n^2s_p^2}\sum_{t=1}^n\Var(\vzeta_t^\top\vzeta_t)
		\le \frac{C\tr(\mR^2)}{ns_p^2}
		\le Cn^{-1},
		\qquad
		D_n=O_{\mathbb P}(n^{-1/2}).
	\end{equation*}
	Let
	\begin{equation*}
		\mathcal F_t=\sigma(\vzeta_1,\ldots,\vzeta_t),
		\qquad
		\bm H_{t-1}=\sum_{s=1}^{t-1}\vzeta_s,
		\qquad
		\bm H_0=\bm0.
	\end{equation*}
	Then
	\begin{equation*}
		W_n=\sum_{t=1}^n\xi_{nt},
		\qquad
		\xi_{nt}=\frac{2}{s_p}
		\left(n^{-1}\bm H_{t-1}+n^{-1/2}\bm\gamma_n\right)^\top\vzeta_t,
	\end{equation*}
	\begin{equation*}
		\Ee(\xi_{nt}\mid\mathcal F_{t-1})=0,
	\end{equation*}
	\begin{align*}
		V_{n,\gamma}
		&=\sum_{t=1}^n\Ee(\xi_{nt}^2\mid\mathcal F_{t-1})\notag\\
		&=\frac4{s_p^2}\sum_{t=1}^n
		\left(n^{-1}\bm H_{t-1}+n^{-1/2}\bm\gamma_n\right)^\top
		\mR
		\left(n^{-1}\bm H_{t-1}+n^{-1/2}\bm\gamma_n\right)\notag\\
		&=V_{1n}+V_{2n}+\Lambda_n,
	\end{align*}
	where
	\begin{equation*}
		V_{1n}=\frac4{n^2s_p^2}\sum_{t=1}^n\bm H_{t-1}^\top\mR\bm H_{t-1},
		\qquad
		V_{2n}=\frac8{n^{3/2}s_p^2}\sum_{t=1}^n\bm H_{t-1}^\top\mR\bm\gamma_n.
	\end{equation*}
	Since $\Ee(\bm H_{t-1}\bm H_{t-1}^\top)=(t-1)\mR$,
	\begin{equation*}
		\Ee V_{1n}=\frac4{n^2s_p^2}\sum_{t=1}^n(t-1)\tr(\mR^2)
		=1-n^{-1}.
	\end{equation*}
	The block quadratic-form calculation in Lemma~\ref{lem:oracle_subg} gives
	\begin{equation*}
		\Var(V_{1n})
		\le Cn^{-2}+C\frac{\tr(\mR^4)}{n\tr^2(\mR^2)}
		\le Cn^{-2}+C(np)^{-1}.
	\end{equation*}
	Moreover,
	\begin{equation*}
		\sum_{t=1}^n\bm H_{t-1}=\sum_{s=1}^{n-1}(n-s)\vzeta_s,
	\end{equation*}
	so
	\begin{align*}
		\Ee V_{2n}^2
		&=\frac{64}{n^3s_p^4}
		\sum_{s=1}^{n-1}(n-s)^2\bm\gamma_n^\top\mR^3\bm\gamma_n\notag\\
		&\le C\frac{\bm\gamma_n^\top\mR^3\bm\gamma_n}{s_p^4}
		\le C\frac{\lambda_{\max}^2(\mR)\bm\gamma_n^\top\mR\bm\gamma_n}{\tr^2(\mR^2)}
		\le C\frac{\Lambda_n}{p}.
	\end{align*}
	Under \eqref{eq:local_alt_condition},
	\begin{equation*}
		V_{n,\gamma}=1+\Lambda_n+O_{\mathbb P}(n^{-1}+p^{-1/2})+o_{\mathbb P}(1)
		=1+\Lambda+o_{\mathbb P}(1).
	\end{equation*}
	The conditional Lindeberg condition is verified as follows.  For every fixed $\varepsilon>0$,
	\begin{align*}
		L_{n,\varepsilon}
		&=\sum_{t=1}^n
		\Ee\{\xi_{nt}^2\mathbf 1(|\xi_{nt}|>\varepsilon)\mid\mathcal F_{t-1}\} \\
		&\le \varepsilon^{-1}
		\sum_{t=1}^n\Ee\{|\xi_{nt}|^3\mid\mathcal F_{t-1}\}.
	\end{align*}
	For $\bm a_{t-1}=n^{-1}\bm H_{t-1}+n^{-1/2}\bm\gamma_n$,
	\begin{equation*}
		\Ee\{|\xi_{nt}|^3\mid\mathcal F_{t-1}\}
		\le
		C s_p^{-3}\{\bm a_{t-1}^\top\mR\bm a_{t-1}\}^{3/2}.
	\end{equation*}
	Using $(a+b)^{3/2}\le 2(a^{3/2}+b^{3/2})$ and $s_p^2=2\tr(\mR^2)\asymp p$,
	\begin{align*}
		\sum_{t=1}^n s_p^{-3}\{\bm a_{t-1}^\top\mR\bm a_{t-1}\}^{3/2}
		&\le
		C\sum_{t=1}^n \left(\frac{t}{n^2}\right)^{3/2}
		+C\sum_{t=1}^n\left(\frac{\bm\gamma_n^\top\mR\bm\gamma_n}{ns_p^2}\right)^{3/2} \\
		&\le Cn^{-1/2}+Cn^{-1/2}\Lambda_n^{3/2}=o_{\mathbb P}(1).
	\end{align*}
	Hence
	\begin{equation*}
		L_{n,\varepsilon}=o_{\mathbb P}(1).
	\end{equation*}
	For fixed $u\in\R$,
	\begin{equation*}
		\Ee\{e^{iu\xi_{nt}}\mid\mathcal F_{t-1}\}
		=1-\frac{u^2}{2}\Ee(\xi_{nt}^2\mid\mathcal F_{t-1})+r_{nt},
	\end{equation*}
	where, by the conditional sub-Gaussian moment bound,
	\begin{equation*}
		|r_{nt}|
		\le C|u|^3\{\Ee(\xi_{nt}^2\mid\mathcal F_{t-1})\}^{3/2}.
	\end{equation*}
	Using $(a+b)^{3/2}\le 2(a^{3/2}+b^{3/2})$,
	\begin{align*}
		&\sum_{t=1}^n\Ee\{\Ee(\xi_{nt}^2\mid\mathcal F_{t-1})\}^{3/2}\notag\\
		&\qquad\le C\sum_{t=1}^n\left(\frac{t}{n^2}\right)^{3/2}
		+C\sum_{t=1}^n\left(\frac{\bm\gamma_n^\top\mR\bm\gamma_n}{ns_p^2}\right)^{3/2}
		\le Cn^{-1/2}+Cn^{-1/2}\Lambda_n^{3/2}.
	\end{align*}
	Thus
	\begin{equation*}
		\sum_{t=1}^n|r_{nt}|=o_{\mathbb P}(1).
	\end{equation*}
	Also,
	\begin{equation*}
		\sum_{t=1}^n\Ee\{\Ee(\xi_{nt}^2\mid\mathcal F_{t-1})^2\}
		\le Cn^{-1}+Cn^{-1}\Lambda_n^2=o(1).
	\end{equation*}
	Iterating conditional expectations gives
	\begin{align*}
		\Ee e^{iuW_n}
		&=\Ee\prod_{t=1}^n
		\left[1-\frac{u^2}{2}\Ee(\xi_{nt}^2\mid\mathcal F_{t-1})+r_{nt}\right]\notag\\
		&=\Ee\exp\left[-\frac{u^2}{2}V_{n,\gamma}+o_{\mathbb P}(1)\right]
		\to \exp\{-u^2(1+\Lambda)/2\}.
	\end{align*}
	Therefore
	\begin{equation*}
		W_n\Rightarrow N(0,1+\Lambda),
		\qquad
		\frac{n\bar\vZ^\top\bar\vZ-p}{s_p}\Rightarrow N(\Delta,1+\Lambda).
	\end{equation*}
	The statistic $T_S$ uses $(n-1)p/(n-3)$ instead of $p$ and the feasible denominator.  Since
	\begin{equation*}
		\frac{(n-1)p/(n-3)-p}{s_p}=O(\sqrt p/n)=o(1),
	\end{equation*}
	\begin{equation*}
		\left|\frac{\{2[\tr(\widehat\mR^2)-p^2/(n-1)]\}^{1/2}}{s_p}-1\right|=o_{\mathbb P}(1),
	\end{equation*}
	and \eqref{eq:replacement_H1_proof} is $o_{\mathbb P}(1)$ under \eqref{eq:local_alt_condition}, Slutsky's theorem yields
	\begin{equation*}
		T_S\Rightarrow N(\Delta,1+\Lambda).
	\end{equation*}
	Consequently,
	\begin{equation*}
		\Pp_{\vmu}(T_S>z_{1-\alpha})
		\to
		1-\Phi\left(\frac{z_{1-\alpha}-\Delta}{\sqrt{1+\Lambda}}\right).
	\end{equation*}
	For diverging alternatives, the exact local normal approximation is not needed.  The martingale variance identity gives
	\begin{equation*}
		\Ee W_n^2=\Ee V_{n,\gamma}=1-n^{-1}+\Lambda_n,
		\qquad
		W_n=O_{\mathbb P}\{(1+\Lambda_n)^{1/2}\}.
	\end{equation*}
	Together with $D_n=O_{\mathbb P}(n^{-1/2})$, \eqref{eq:replacement_H1_proof}, and $\eta_n\to0$,
	\begin{equation*}
		T_S=\Delta_n+O_{\mathbb P}\{(1+\Lambda_n)^{1/2}\}+o_{\mathbb P}(\Delta_n).
	\end{equation*}
	Hence \eqref{eq:sum_consistency_condition} implies
	\begin{equation*}
		\Pp_{\vmu}(T_S>z_{1-\alpha})\to1.
	\end{equation*}
	Furthermore,
	\begin{equation*}
		\Lambda_n
		=\frac{2\bm\gamma_n^\top\mR\bm\gamma_n}{\tr(\mR^2)}
		\le \frac{2\lambda_{\max}(\mR)\norm{\bm\gamma_n}_2^2}{\tr(\mR^2)}
		\le C\frac{\Delta_n}{\sqrt p},
	\end{equation*}
	so $\Delta_n\to\infty$ implies \eqref{eq:sum_consistency_condition}.
	
	For the max statistic, for $x\ge0$ and $C$ large enough,
	\begin{equation*}
		\Pp\left(\max_i\abs{X_{ni}}>C\sqrt{\log p}+x\right)
		\le \sum_{i=1}^p2\exp\{-c(C\sqrt{\log p}+x)^2\}
		\le 2\exp(-cx^2).
	\end{equation*}
	Thus $\max_i|X_{ni}|=O_{\mathbb P}(\sqrt{\log p})$.  Also,
	\begin{equation*}
		\sqrt n\norm{\bar{\widehat{\vZ}}-\bar{\vZ}}_\infty
		=O_{\mathbb P}(\eta_n\sqrt{\log p})+O_{\mathbb P}(\eta_n m_\infty),
	\end{equation*}
	where the second term is the deterministic mean leakage through $\widehat\mP-\mP$.  Since $\eta_n\to0$, if $m_\infty^2-2\log p\to\infty$, then for $i_0$ satisfying $|\gamma_{n,i_0}|=m_\infty$,
	\begin{align*}
		T_M
		&\ge \left(m_\infty-|X_{n,i_0}|-o_{\mathbb P}(m_\infty)-O_{\mathbb P}(\sqrt{\log p}\eta_n)\right)^2
		-2\log p+\log\log p\notag\\
		&\xrightarrow{\mathbb P}+\infty.
	\end{align*}
	Hence the max test has power tending to one.  If either the max condition or \eqref{eq:sum_consistency_condition} holds, at least one of $\widehat p_M$ and $\widehat p_S$ converges to zero in probability, and
	\begin{equation*}
		T_C=\frac12\tan\{\pi(1/2-\widehat p_M)\}
		+\frac12\tan\{\pi(1/2-\widehat p_S)\}
		\xrightarrow{\mathbb P}+\infty.
	\end{equation*}
\end{proof}

\subsection{Proof of Theorem~\ref{thm:bootstrap}}

\begin{proof}
	Let $\Pp^*(\cdot)=\Pp(\cdot\mid\vY_1,\ldots,\vY_n,\widetilde\mA_1,\ldots,\widetilde\mA_B)$.  We give the argument for one bootstrap replication and suppress the replication index.  Write
	\begin{equation*}
		\widetilde\mP=\widetilde\mA\widetilde\mA^\top,
		\qquad
		\widetilde\mQ=\mI_p-\widetilde\mP,
		\qquad
		\vY_t^*=\widetilde{\mA}\widehat{\vf}_t+\widetilde{\veps}_{I_t},
		\qquad
		\Pp^*(I_t=j)=n^{-1}.
	\end{equation*}
	Since the fitted factors and centered residuals satisfy
	\begin{equation*}
		\sum_{t=1}^n\widehat{\vf}_t=\bm0,
		\qquad
		\sum_{t=1}^n\widetilde{\veps}_t=\bm0,
	\end{equation*}
	the bootstrap null mean is zero.  The oracle projected bootstrap residuals are
	\begin{equation*}
		\veps_t^\circ=\widetilde{\mQ}\widetilde{\veps}_t,
		\qquad
		\widehat{\mOmega}_\circ=n^{-1}\sum_{t=1}^n\veps_t^\circ\veps_t^{\circ\top},
		\qquad
		\widehat{\mD}_\circ=\diag(\widehat{\mOmega}_\circ),
		\qquad
		\widehat{\vzeta}_t^\circ=\widehat{\mD}_\circ^{-1/2}\veps_t^\circ,
	\end{equation*}
	with
	\begin{equation*}
		\widehat{\mR}_\circ=n^{-1}\sum_{t=1}^n\widehat{\vzeta}_t^\circ\widehat{\vzeta}_t^{\circ\top}.
	\end{equation*}
	Then
	\begin{equation*}
		\Ee^*(\widehat{\vzeta}_{I_t}^\circ)=\bm0,
		\qquad
		\Cov^*(\widehat{\vzeta}_{I_t}^\circ)=\widehat{\mR}_\circ,
		\qquad
		\Cov^*(\widehat{\vzeta}_{I_t}^\circ,\widehat{\vzeta}_{I_s}^\circ)=\bm0\quad(t\ne s).
	\end{equation*}
	The residual trace and denominator bounds established in Lemmas~\ref{lem:resid_consequences} and \ref{lem:denom_trace}, applied conditionally to the empirical residual distribution, give
	\begin{equation*}
		\frac{\tr(\widehat\mR_\circ^2)-\tr(\mR^2)}{\tr(\mR^2)}=o_{\mathbb P}(1),
		\qquad
		\lambda_{\max}(\widehat\mR_\circ)\le C+o_{\mathbb P}(1),
		\qquad
		\max_i\sum_{j=1}^p\widehat R_{\circ,ij}^2\le C+o_{\mathbb P}(1).
	\end{equation*}
	Moreover, for every fixed $q\ge1$,
	\begin{equation*}
		\max_{1\le i\le p}\Ee^*|\widehat\zeta_{I_t,i}^\circ|^q
		\le n^{-1}\sum_{s=1}^n\max_i|\widehat\zeta_{s,i}^\circ|^q
		=O_{\mathbb P}\{(\log p)^{q/2}\},
	\end{equation*}
	\begin{equation*}
		\max_{t,i}|\widehat\zeta_{t,i}^\circ|=O_{\mathbb P}(\sqrt{\log(np)}),
		\qquad
		\max_i n^{-1}\sum_{t=1}^n(\widehat\zeta_{t,i}^\circ)^4=O_{\mathbb P}(1).
	\end{equation*}
	Set
	\begin{equation*}
		\bm X_n^*=n^{-1/2}\sum_{t=1}^n\widehat{\vzeta}_{I_t}^\circ.
	\end{equation*}
	The conditional max Gaussian approximation yields
	\begin{align*}
		&\sup_x\left|\Pp^*\left(\max_i |X_{n,i}^*|^2-2\log p+\log\log p\le x\right)
		-\Pp\left(\max_i G_{\circ,i}^2-2\log p+\log\log p\le x\mid\widehat\mR_\circ\right)\right|\\
		&\qquad\le C\left(\frac{\log^7p}{n}\right)^{1/6}
		+C\left(n^{-1}\sum_{t=1}^n\max_i|\widehat\zeta_{t,i}^\circ|^3\right)^{1/3}\frac{\log^{2/3}p}{n^{1/6}}
		=o_{\mathbb P}(1),
	\end{align*}
	where $\bm G_\circ\mid\widehat\mR_\circ\sim N(\bm0,\widehat\mR_\circ)$.  The empirical version of the extreme-value regularity condition and the preceding trace bounds imply
	\begin{equation*}
		\sup_x\left|\Pp\left(\max_i G_{\circ,i}^2-2\log p+\log\log p\le x\mid\widehat\mR_\circ\right)-F_G(x)\right|=o_{\mathbb P}(1).
	\end{equation*}
	For the oracle quadratic component,
	\begin{equation*}
		(\bm X_n^*)^\top\bm X_n^*-p
		=\frac2n\sum_{1\le s<t\le n}(\widehat{\vzeta}_{I_s}^\circ)^\top\widehat{\vzeta}_{I_t}^\circ
		+\frac1n\sum_{t=1}^n\{(\widehat{\vzeta}_{I_t}^\circ)^\top\widehat{\vzeta}_{I_t}^\circ-p\}.
	\end{equation*}
	The diagonal term satisfies
	\begin{equation*}
		\Ee^*\left[
		\frac1n\sum_{t=1}^n\{(\widehat{\vzeta}_{I_t}^\circ)^\top\widehat{\vzeta}_{I_t}^\circ-p\}
		\right]^2
		\le Cn^{-1}\tr(\widehat\mR_\circ^2),
	\end{equation*}
	and hence is $O_{\mathbb P}^*(n^{-1/2})$ after normalization by $\{2\tr(\widehat\mR_\circ^2)\}^{1/2}$.  For the cross term define
	\begin{equation*}
		\mathcal F_t^*=\sigma(I_1,\ldots,I_t),
		\qquad
		\bm H_{t-1}^*=\sum_{s=1}^{t-1}\widehat{\vzeta}_{I_s}^\circ,
		\qquad
		\xi_{nt}^*=\frac{2}{n\{2\tr(\widehat{\mR}_\circ^2)\}^{1/2}}(\bm H_{t-1}^*)^\top\widehat{\vzeta}_{I_t}^\circ.
	\end{equation*}
	Then
	\begin{equation*}
		\Ee^*(\xi_{nt}^*\mid\mathcal F_{t-1}^*)=0,
	\end{equation*}
	\begin{align*}
		\sum_{t=1}^n\Ee^*((\xi_{nt}^*)^2\mid\mathcal F_{t-1}^*)
		&=\frac{2}{n^2\tr(\widehat\mR_\circ^2)}
		\sum_{t=1}^n(\bm H_{t-1}^*)^\top\widehat\mR_\circ\bm H_{t-1}^* \\
		&=1+O_{\mathbb P}^*(n^{-1/2}+p^{-1/2}),
	\end{align*}
	\begin{equation*}
		\sum_{t=1}^n\Ee^*\{|\xi_{nt}^*|^3\mid\mathcal F_{t-1}^*\}=O_{\mathbb P}^*(n^{-1/2}).
	\end{equation*}
	Thus, for every fixed $\epsilon>0$,
	\begin{equation*}
		\sum_{t=1}^n\Ee^*\{(\xi_{nt}^*)^2\mathbf 1(|\xi_{nt}^*|>\epsilon)\mid\mathcal F_{t-1}^*\}
		\le \epsilon^{-1}\sum_{t=1}^n\Ee^*\{|\xi_{nt}^*|^3\mid\mathcal F_{t-1}^*\}=o_{\mathbb P}^*(1),
	\end{equation*}
	and the conditional martingale central limit theorem gives
	\begin{equation*}
		\sup_y\left|\Pp^*\left(
		\frac{(\bm X_n^*)^\top\bm X_n^*-p}{\{2\tr(\widehat\mR_\circ^2)\}^{1/2}}\le y\right)-\Phi(y)\right|=o_{\mathbb P}(1).
	\end{equation*}
	It remains to relate the oracle bootstrap statistics to the feasible bootstrap statistics.  Let $\widehat\mP^*$ and $\widehat\mQ^*=\mI_p-\widehat\mP^*$ be the loading-space estimator and projector computed from the bootstrap sample, and let
	\begin{equation*}
		\Delta_P^*=\widehat\mP^*-\widetilde\mP,
		\qquad
		\widehat\mD^*=\diag\left\{n^{-1}\sum_{t=1}^n\widehat\mQ^*\vY_t^*\vY_t^{*\top}\widehat\mQ^*\right\}.
	\end{equation*}
	On the bootstrap selection event, the conditional eigenspace expansion gives
	\begin{equation*}
		\|\Delta_P^*\|=O_{\mathbb P}^*(\eta_n),
		\qquad
		\max_i\|\ve_i^\top\Delta_P^*\|_2=O_{\mathbb P}^*(\sqrt{\log p/p}\,\eta_n),
		\qquad
		\eta_n=p^\delta\left(\frac{\log p}{n}\right)^{1/2}+p^\delta\frac{\log p}{n}.
	\end{equation*}
	Because $n^{-1}\sum_t\widehat{\vf}_t=\bm0$,
	\begin{equation*}
		\sqrt n\bar\vY^*=n^{-1/2}\sum_{t=1}^n\widetilde{\veps}_{I_t},
		\qquad
		\widetilde\mQ\sqrt n\bar\vY^*=\widehat\mD_\circ^{1/2}\bm X_n^*.
	\end{equation*}
	The max replacement follows from
	\begin{align*}
		\max_i\left|\sqrt n\bar{\widehat Z}_{i}^{*}-X_{n,i}^*\right|
		&\le
		\max_i\left|\ve_i^\top(\widehat\mD^{*-1/2}-\widehat\mD_\circ^{-1/2})\widetilde\mQ\sqrt n\bar\vY^*\right|\\
		&\quad+
		\max_i\left|\ve_i^\top\widehat\mD^{*-1/2}(\widehat\mQ^*-\widetilde\mQ)\sqrt n\bar\vY^*\right|\\
		&\le
		O_{\mathbb P}^*(\delta_{D,n}^*)\max_i|X_{n,i}^*|
		+O_{\mathbb P}^*(\sqrt{\log p/p}\,\eta_n)\left\|n^{-1/2}\sum_{t=1}^n\widetilde{\veps}_{I_t}\right\|_2,
	\end{align*}
	where
	\begin{equation*}
		\delta_{D,n}^*=\max_i\left|\widehat D_i^*/\widehat D_{\circ,i}-1\right|
		=O_{\mathbb P}^*(\eta_n+n^{-1/2}).
	\end{equation*}
	Since
	\begin{equation*}
		\max_i|X_{n,i}^*|=O_{\mathbb P}^*(\sqrt{\log p}),
		\qquad
		\left\|n^{-1/2}\sum_{t=1}^n\widetilde{\veps}_{I_t}\right\|_2=O_{\mathbb P}^*(\sqrt p),
	\end{equation*}
	we obtain
	\begin{equation*}
		\max_i\left|\sqrt n\bar{\widehat Z}_{i}^{*}-X_{n,i}^*\right|
		=O_{\mathbb P}^*\{(\eta_n+n^{-1/2})\sqrt{\log p}\}.
	\end{equation*}
	Under the polynomial growth conditions this is $o_{\mathbb P}^*((\log p)^{-1/2})$.
	For the quadratic statistic, the conditional version of Lemmas~\ref{lem:strong_projection_drift} and \ref{lem:replacement} gives the explicit expansion
	\begin{equation*}
		n\bar{\widehat{\vZ}}^{*\top}\bar{\widehat{\vZ}}^{*}-(\bm X_n^*)^\top\bm X_n^*
		=\mathcal B_{P,n}^*+\mathcal R_{P,n}^*,
	\end{equation*}
	where, in the strong-factor case,
	\begin{equation*}
		\frac{|\mathcal B_{P,n}^*|}{\{\tr(\widehat\mR_\circ^2)\}^{1/2}}=O_{\mathbb P}^*(\sqrt p/n),
		\qquad
		\frac{|\mathcal R_{P,n}^*|}{\{\tr(\widehat\mR_\circ^2)\}^{1/2}}
		=O_{\mathbb P}^*\left(n^{-1/2}+\frac{\sqrt p}{n}+\frac{\log p}{n}\right),
	\end{equation*}
	and, in the weak-factor case,
	\begin{equation*}
		\frac{|\mathcal B_{P,n}^*+\mathcal R_{P,n}^*|}{\{\tr(\widehat\mR_\circ^2)\}^{1/2}}
		=O_{\mathbb P}^*\left(\frac{p^{\delta/2}}{\sqrt n}+p^{1/2+2\delta}\frac{\log p}{n}+\frac{\sqrt p}{n}\right).
	\end{equation*}
	Hence
	\begin{equation*}
		\frac{\left|n\bar{\widehat{\vZ}}^{*\top}\bar{\widehat{\vZ}}^{*}-(\bm X_n^*)^\top\bm X_n^*\right|}{\{\tr(\widehat{\mR}_\circ^2)\}^{1/2}}
		=o_{\mathbb P}^*(1).
	\end{equation*}
	The same diagonal normalization and leave-two-out trace expansion as Lemma~\ref{lem:denom_trace}, conditionally applied to $\widehat{\vzeta}_{I_t}^\circ$, yields
	\begin{equation*}
		\frac{\tr((\widehat\mR^*)^2)-p^2/(n-1)-\tr(\widehat\mR_\circ^2)}{\tr(\widehat\mR_\circ^2)}=o_{\mathbb P}^*(1).
	\end{equation*}
	Combining the oracle bootstrap limits with the feasible-oracle replacements gives
	\begin{align*}
		\sup_x\abs{\Pp^*(T_M^*\le x)-F_G(x)}&=o_{\mathbb P}(1),\\
		\sup_y\abs{\Pp^*(T_S^*\le y)-\Phi(y)}&=o_{\mathbb P}(1),\\
		\sup_{x,y}\abs{\Pp^*(T_M^*\le x,T_S^*\le y)-F_G(x)\Phi(y)}&=o_{\mathbb P}(1).
	\end{align*}
	For $B\to\infty$,
	\begin{equation*}
		\widehat p_M^{\rm boot}-\Pp^*(T_M^*\ge T_M)=O_{\mathbb P}^*(B^{-1/2}),
		\qquad
		\widehat p_S^{\rm boot}-\Pp^*(T_S^*\ge T_S)=O_{\mathbb P}^*(B^{-1/2}).
	\end{equation*}
	The preceding conditional distributional convergence gives
	\begin{equation*}
		\widehat p_M^{\rm boot}-p_{M,0}=o_{\mathbb P}(1),
		\qquad
		\widehat p_S^{\rm boot}-p_{S,0}=o_{\mathbb P}(1),
	\end{equation*}
	where $p_{M,0}$ and $p_{S,0}$ are the corresponding first-order null $p$-values.  The Cauchy statistic follows by applying
	\begin{equation*}
		(p_M,p_S)\mapsto
		\frac12\tan\{\pi(1/2-p_M)\}+\frac12\tan\{\pi(1/2-p_S)\}
	\end{equation*}
	on compact subsets away from $0$ and $1$, while the endpoint neighborhoods have probability $o(1)$ under the null calibration.
\end{proof}

\bibliographystyle{cas-model2-names}
\bibliography{references}

\end{document}